\documentclass{llncs}
\usepackage[USenglish]{babel}
\usepackage{collect}
\usepackage{amsmath}
\usepackage{amssymb}
\usepackage{stmaryrd}
\usepackage{tikz}
\usetikzlibrary{arrows,fit,decorations.pathreplacing,decorations.pathmorphing,positioning,patterns,shapes}
\usepackage{complexity}
\usepackage{booktabs}
\usepackage{xspace}
\usepackage{hyperref}
\usepackage[olditem,oldenum]{paralist}
\usepackage{color,wrapfig}
\usepackage{macros}
\usepackage{caption}
\renewcommand\Pl{\Agt}

\usepackage[ruled]{algorithm2e}
\newcommand\longVersion[1]{}
\newcommand\draft[1]{}

\newenvironment{myitemize}{
  \begin{inparaenum}}{
  \end{inparaenum}}

\newenvironment{mydefinition}[1][]{\paragraph{\bf #1.}}{}

\definecollection{appendix-randomised}
\definecollection{appendix-suspect}
\definecollection{appendix-algorithm}
\definecollection{appendix-mean-payoff}
\definecollection{appendix-fixed-coalition}
\definecollection{appendix-hardness}

\title{Robust Equilibria in Mean-Payoff Games}

\author{Romain Brenguier\thanks{Work supported by ERC Starting Grant inVEST (279499) and EPSRC grant EP/M023656/1.}}
\institute{University of Oxford, UK}

\newcommand\mylabel[2]{\label{#2}}
\newcommand\mysetcounter[2]{}
\newcommand\mysubsubsection[1]{\smallskip \noindent \textbf{#1} }

\begin{document}

\maketitle

\begin{abstract}
  We study the problem of finding robust equilibria in multiplayer concurrent games with mean payoff objectives.
  A $(k,t)$-robust equilibrium is a strategy profile such that no coalition of size $k$ can improve the payoff of one its member by deviating, and no coalition of size $t$ can decrease the payoff of other players.
  While deciding whether there exists such equilibria is undecidable in general, we suggest algorithms for two meaningful restrictions on the complexity of strategies.
  The first restriction concerns memory.
  We show that we can reduce the problem of the existence of a memoryless robust equilibrium to a formula in the (existential) theory of reals.
  The second restriction concerns randomisation.
  We suggest a general transformation from multiplayer games to two-player games such that pure equilibria in the first game correspond to winning strategies in the second one.
  Thanks to this transformation, we show that the existence of robust equilibria can be decided in polynomial space, and that the decision problem is \PSPACE-complete.
\end{abstract}


\section{Introduction}
Games are intensively used in computer science to model interactions in computerised systems.
Two player antagonistic games have been successfully used for the synthesis of reactive systems.
In this context, the opponent acts as a hostile environment, and winning strategies provide controllers that ensure correctness of the system under any scenario.
In order to model complex systems in which several rational entities interact, multiplayer concurrent games come into the picture.
Correctness of the strategies can be specified with different solution concepts, which describe formally what is a ``good'' strategy.
In game theory, the fundamental solution concept is Nash equilibrium~\cite{nash50}, where no player can benefit from changing its own strategy.
The notion of robust equilibria refines Nash equilibria in two ways:
\begin{myitemize}
\item
  a robust equilibrium is \emph{resilient}, \ie when a ``small'' coalition of player changes its strategy, it can not improve the payoff of one of its participants;
\item
  it is \emph{immune}, \ie when a ``small'' coalition changes its strategy, it will not lower the payoff of the non-deviating players.
\end{myitemize}
The size of what is considered a small coalition is determined by a bound $k$ for resilience and another $t$ for immunity.
When a strategy is both $k$-resilient and $t$-immune, it is called a $(k,t)$-robust equilibrium.
We also generalise this concept to the notion $(k,t,r)$-robust equilibrium, where if $t$ players are deviating, the others should not have their payoff decrease by more than $r$.

\paragraph*{Example} 
In the design of network protocols, when many users are interacting, coalitions can easily be formed and resilient strategies are necessary to avoid deviation.
It is also likely that some clients are faulty and begin to behave unexpectedly, hence the need for immune strategies.

As an example, consider a program for a server that distributes files, of which a part is represented in \figurename~\ref{fig:program}.
The functions \texttt{listen} and \texttt{send\_files} will be run in parallel by the server.
Some choices in the design of these functions have not been fixed yet and we wish to analyse the robustness of the different alternatives.

This program uses a table~\texttt{clients} to keep track of the clients which are connected.
Notice that the table has fixed size 2, which means that if 3 clients try to connect at the same time, one of them may have its socket overwritten in the table and will have to reconnect later to get the file.
We want to know what strategy the clients should use and how robust the protocol will be: can clients exploit the protocol to get their files faster than the normal usage, and how will the performance of the over clients be affected.

We consider different strategies to chose between the possible alternatives in the program of \figurename~\ref{fig:program}.
The strategy that chooses alternatives $1$ and $3$ does not give $1$-resilient equilibria even for just two clients:
Player~$1$ can always reconnect just after its socket was closed, so that \texttt{clients[0]} points to player~$1$ once again.
In this way, he can deviate from any profile to never have to wait for the file.
Since the second player could do the same thing, no profile is $1$-resilient (nor $1$-immune).
For the same reasons, the strategy $2, 3$ does not give $1$-resilient equilibria.
The strategy $1, 4$ does not give $1$-resilient equilibria either, since player~$1$ can launch a new connection after player~$2$ to overwrite \texttt{clients[1]}.

The strategy $2, 4$ is the one that may give the best solution. 
We modelled the interaction of this program as a concurrent game for a situation with 2 potential clients in \figurename~\ref{fig:game-from-program}.
The labels represent the content of the table \texttt{clients}: $0$ means no connection, $1$ means connected with player~$1$ and $2$ connected with player~$2$; and the instruction that the function \texttt{send\_files} is executing.
Because the game is quite big we represent only the part where player $1$ connects before player~$2$ and the rest of the graph can be deduced by symmetry.
The actions of the players are either to wait (action {\texttt w}) or to connect (action {\texttt{ch}} or {\texttt{ct}}).
Symbol $\ast$ means any possible action.
Note that in the case both player try to connect at the same time we simulate a matching penny game in order to determine which one will be treated first, this is the reason why we have two different possible actions to connect (\texttt{ch} for ``head'' and \texttt{ct} for ``tail'').
Clients have a positive reward when we send them the file they requested,
this corresponds for Player~$i$ to states labelled with \texttt{send(i)}.

If both clients try to connect in the same slot, we use a matching penny game to decide which request was the first to arrive.
For a general method to transform a game with continuous time into a concurrent game, see~\cite[Chapter~6]{brenguier12}.

\begin{figure}[htb]
\begin{minipage}{0.5\textwidth}
{\scriptsize
  \begin{verbatim}
clients = new socket[2];

void listen() {
  while(true) {
    Socket socket = serverSocket.accept();
    if(clients[0].isConnected())
    //// Two possible alternatives:
    | 1) clients[1] = socket;
    | 2) if(socket.remoteSocketAddress() 
          != clients[0].remoteSocketAddress())
    |      clients[1] = socket;
    else
      clients[0] = socket;
  } }

void send_files() {
  while(true) {
    if(clients[0].isConnected()) {
      send(clients[0]);
      clients[0].close();
    } 
    //// Two possible alternatives :
    | 3) else if(clients[1].isConnected()) {  
    | 4) if(clients[1].isConnected()) {
      send(clients[1]);                  
      clients[1].close();                
} } }
  \end{verbatim}
}
  \caption{Example of a server program.}\label{fig:program}
\end{minipage}
  \begin{minipage}{0.5\textwidth}
  \centering{\scriptsize
    \begin{tikzpicture}[yscale=0.8,xscale=0.45]
      \draw (0,0) node[draw,rounded corners=3mm,minimum size=6mm] (S00) { $[0,0]$ };
      \draw[-latex'] (-2,0) -- (S00);

      \draw (5,3) node[draw,rounded corners=3mm,minimum size=6mm,text width=1cm,text centered] (T10) { $[1,0]$ \texttt{send(1)} \texttt{close(1)}};
      \draw (5,1) node[draw,rounded corners=3mm,minimum size=6mm,text width=1cm,text centered] (T12) { $[1,2]$ \texttt{send(1)} \texttt{close(1)}};
      \draw (5,-3) node[draw,rounded corners=3mm,minimum size=6mm,text width=1cm,text centered] (T20) { $[2,0]$ \texttt{send(2)}};
      \draw (0,-3) node[draw,rounded corners=3mm,minimum size=6mm,text width=1cm,text centered] (T21) { $[2,1]$ \texttt{send(2)}};

      \draw[-latex'] (S00) -- node[above,sloped] {\texttt{ch,w}} node[below,sloped]{\texttt{ct,w}} (T10);
      \draw[-latex'] (S00) -- node[above,sloped] {\texttt{ch,ct}} node[below,sloped]{\texttt{ct,ch}} (T12);
      \draw[-latex'] (S00) -- node[above,sloped] {\texttt{w,ct}} node[below,sloped]{\texttt{w,ch}} (T20);
      \draw[-latex'] (S00) -- node[above,sloped] {\texttt{ch,ch}} node[below,sloped]{\texttt{ct,ct}} (T21);
      \draw[-latex'] (S00) .. controls +(0,-1) and +(-1,-1) .. node[below left,sloped] {\texttt{w,w}} (S00);

      \draw(T21.-90) node[below] {\dots};
      \draw(T20.-90) node[below] {\dots};
     \draw (10,0) node[draw,rounded corners=3mm,minimum size=6mm,text width=1cm,text centered] (V12) {$[0,2]$};

      \draw[-latex'] (T12) -- node[above,sloped] {$\ast$, $\ast$} (V12);
      \draw[-latex'] (T10) -- node[above,sloped] {$\ast$, \texttt{ch}} node[below,sloped]{\texttt{$\ast$, ct}} (V12);

      \draw (5,-1) node[draw,rounded corners=3mm,minimum size=6mm,text width=1cm,text centered] (W02) { $[0,2]$ \texttt{send(2)} \texttt{close(2)}};
      \draw (10,3) node[draw,rounded corners=3mm,minimum size=6mm,text width=1cm,text centered] (W12) { $[1,2]$ \texttt{send(2)} \texttt{close(2)}};
      \draw (8.5,-2) node[draw,rounded corners=3mm,minimum size=6mm,text width=1cm,text centered] (W22) { $[2,2]$ \texttt{send(2)} \texttt{close(2)}};
      \draw (11,-3) node[draw,rounded corners=3mm,minimum size=6mm,text width=1cm,text centered] (W21) { $[2,1]$ \texttt{send(1)} \texttt{close(1)}};

      \draw[-latex'] (V12) -- node[above,sloped] {\texttt{w,w}} (W02);
      \draw[-latex'] (V12) -- node[right,pos=0.2] {\texttt{ch,w}} node[right,pos=0.4]{\texttt{ct,w}} node[right,pos=0.6]{\texttt{ch,ct}} node[right,pos=0.8]{\texttt{ct,ch}} (W12);
      \draw[-latex'] (V12) -- node[left,pos=0.3] {\texttt{w,ct}} node[left,pos=0.7]{\texttt{w,ch}} (W22);
      \draw[-latex',rounded corners=3mm] (V12) -- node[pos=0.5,right] {\texttt{ch,ch}} node[right,pos=0.7]{\texttt{ct,ct}} (W21);

      \draw[-latex',rounded corners=3mm] (T10) -| node[above,pos=0.4] {$\ast$, \texttt{w}} (S00);
      \draw[-latex'] (W02) -- (S00);
      \draw[-latex'] (W12) -- (T10);
      \draw[-latex'] (W22) -- (T20);
      \draw[-latex'] (W21) -- (T20);

    \end{tikzpicture}
  }
  \caption{Example of a concurrent game generated from the program of \figurename~\ref{fig:program}.
  }\label{fig:game-from-program}
  \end{minipage}
\end{figure}

\paragraph*{Related works and comparison with Nash equilibria and secure equilibria}
Other solution concepts have been proposed as concepts for synthesis of distributed systems,
in particular Nash equilibrium~\cite{ummels2009complexity,ummels2011complexity,BBMU12}, subgame perfect equilibria~\cite{Ummels08,brihaye2010}, and secure equilibria~\cite{CHJ05}.
A subgame perfect equilibria is a particular kind of Nash equilibria, where at any point in the game, if we forget the history the players are still playing a Nash equilibrium.
In a secure equilibria, we ask that no player can benefit or keep the same reward while reducing the payoff of other players by changing its own strategy.
However these concepts present two weaknesses:
\begin{myitemize}
\item
  There is no guarantee when two (or more) users deviate together.
  It can happen on a network that the same person controls several devices (a laptop and a phone for instance) and can then coordinate there behaviour.
  In that case, the devices would be considered as different agents and Nash equilibria offers no guarantee.
\item
  When a deviation occurs, the strategies of the equilibrium can punish the deviating user without any regard for payoffs of the others.
  This can result in a situation where, because of a faulty device, the protocol is totally blocked.
\end{myitemize}
By comparison, finding resilient equilibria with $k$ greater than $1$, ensures that clients have no interest in forming coalitions (up to size $k$), and finding immune equilibria with $t$ greater than $0$ ensures that other clients will not suffer from some agents (up to $t$) behaving differently from what was expected.

Note that the concept of robust equilibria for games with LTL objectives is expressible in logics such as strategy logic~\cite{CHP10} or $\text{ATL}^*$~\cite{AHK02}.
However, satisfiability in these logic is difficult: it is 2\ComplexityFont{EXPTIME}-complete for ATL$^*$ and undecidable for strategy logic in general (2\ComplexityFont{EXPTIME}-complete fragments exist~\cite{MMPV12}).
Moreover, these logics cannot express equilibria in quantitative games such as mean-payoff.


\paragraph*{Contributions}
In this paper, we study the problem of finding robust equilibria in multiplayer concurrent games.
This problem is undecidable in general (see Section~\ref{sec:undec}).
In Section~\ref{sec:randomised}, we show that if we look for stationary (but randomised) strategies, then the problem can be decided using the theory of reals.
We then turn to the case of pure (but memoryful) strategies.
In Section~\ref{sec:suspect}, we describe a generic transformation from multiplayer games to two-player games.
The resulting two-player game is called the \newdef{deviator game}.
We show that pure equilibria in the original game correspond to winning strategies in the second one.
In Section~\ref{sec:mean-payoff}, we study quantitative games with mean-payoff objectives.
We show that the game obtained by our transformation is equivalent to a multidimensional mean-payoff game.
We then show that this can be reduced to a value problem with linear constraints in multidimensional mean-payoff games.
We show that this can be solved in polynomial space, by making use of the structure of the deviator game.
In Section~\ref{sec:hardness}, we prove the matching lower bound which shows the robustness problem is \PSPACE-complete.
Due to space constraints, most proofs have been omitted from this paper; they can be found in the appendix.

\section{Definitions}

\subsection{Weighted concurrent games}

We study concurrent game as defined in~\cite{AHK02} with the addition of weights on the edges.

\begin{mydefinition}[Concurrent games]
  A \newdef{weighted concurrent game} (or simply a \newdef{game})~$\calG$ is a tuple $\langle \Stat, s_0,$ $\Pl, \Act, \Tab, (w_A)_{A\in\Agt} \rangle$, where:
  \begin{inparaitem}[]
  \item $\Stat$ is a finite set of \newdef{states} and $s_0 \in \Stat$ is the \newdef{initial state};
  \item $\Pl$ is a finite set of \newdef{players}; 
  \item $\Act$ is a finite set of \newdef{actions}; a tuple $(\shortAct_A)_{A\in\Pl}$ containing one action~$\shortAct_A$ for each player~$A$ is called a \newdef{move};
  \item $\Tab : \Stat \times \Act^\Pl \rightarrow \Stat$ 
    is the \newdef{transition function}, it associates with a given state and a given move, the resulting state;
  \item for each player $A \in \Agt$, $w_A \colon \Stat \mapsto \Z$ is a \newdef{weight function} which assigns to each agent an integer weight.
  \end{inparaitem}
\end{mydefinition}

In a game~$\calG$, whenever we arrive at a state~$\stat$, the players simultaneously select an action.
This results in a move~$\shortAct_\Pl$; the next state of the game is then $\Tab(\stat,a_\Pl)$.
This process starts from $s_0$ and is repeated to form an infinite sequence of states.

An example of a game is given in \figurename~\ref{fig:game-from-program}. 
It models the interaction of two clients $A_1$ and $A_2$ with the program presented in the introduction.
The weight functions for this game are given by $w_{A_1} = 1$ in states labelled by \texttt{send(1)} and $w_{A_1} = 0$ elsewhere, similarly $w_{A_2} = 1$ in states labelled by \texttt{send(2)}.

\begin{mydefinition}[History and plays]
A \newdef{history} of the game~$\calG$ is a finite sequence of states and moves ending with a state, \ie an word in $(\Stat\cdot \Act^\Pl)^*\cdot \Stat$.
We write $h_i$ the $i$-th state of $h$, starting from $0$, and $\act_i(h)$ its $i$-th move, thus $h= h_0 \cdot \act_0(h) \cdot h_1 \cdots \act_{n-1}(h) \cdot h_n$.
The length $|h|$ of such a history is $n+1$.
We write $\last(h)$ the last state of $h$, \ie $h_{|h|-1}$.
A \newdef{play}~$\rho$ is an infinite sequence of states and moves, 
\ie an element of $(\Stat \cdot \Act^\Agt)^\omega$.
We write $\rho_{\le n}$ for the prefix of $\rho$ of length $n+1$, \ie the history $\rho_0 \cdot \act_0(\rho) \cdots \act_{n-1}(\rho) \cdot \rho_n$.

The \newdef{mean-payoff} of weight~$w$ along a play~$\rho$ is the average of the weights along the play:
\(\MP_w(\rho) = \liminf_{n \rightarrow \infty} \frac{1}{n} \sum_{0 \le k\le n} w(\rho_{k}). \)
The \newdef{payoff} for agent $A\in\Agt$ of a play~$\rho$ is the mean-payoff of the corresponding weight:
$\payoff_A(\rho) = \MP_{w_A}(\rho)$.
Note that it only depends on the sequence of states, and not on the sequence of moves.
The \newdef{payoff vector} of the run $\rho$ is the vector $p\in \mathbb{R}^\Agt$ such that for all players~$A \in \Agt$, $p_A = \payoff_A(\rho)$; we simply write $\payoff(\rho)$ for this vector.
\end{mydefinition}


\begin{mydefinition}[Strategies]
  Let~$\calG$ be a game, and~$A\in\Agt$.    
  A \newdef{strategy} for player~$A$ maps histories to probability distributions over actions.
  Formally, a strategy is a function $\sigma_A\colon \histories \to \calD(\Act)$, where $\calD(\Act)$ is the set of probability distributions over $\Act$.
  For an action $a \in \Act$, we write $\sigma_A(a \mid h)$ the probability assigned to $a$ by the distribution $\sigma(h)$.
  A \newdef{coalition}~$C\subseteq \Pl$ is a set of players, its size is the number of players it contains and we write it $|C|$.
  A~strategy~$\sigma_C = (\sigma_A)_{A\in C}$ for a coalition~$C\subseteq \Agt$ is a tuple of strategies, one for each player in~$C$.
  We write $\sigma_{-C}$ for a strategy of coalition $\Agt\setminus C$.
  A~\newdef{strategy profile} is a strategy for~$\Pl$. 
  We will write $(\sigma_{- C},\sigma'_C)$ for the strategy profile~$\sigma''_\Agt$ such that if $A \in C$ then $\sigma''_A = \sigma'_A$ and otherwise $\sigma''_A = \sigma_A$.
  We~write $\Strat_\calG(C)$ for the set of strategies of coalition~$C$.
  A strategy $\sigma_A$ for player~$A$ is said \newdef{deterministic} if it does not use randomisation: for all histories~$h$ there is an action $a$ such that $\sigma(a \mid h) = 1$.
  A strategy $\sigma_A$ for player~$A$ is said \newdef{stationary} if it depends only on the last state of the history: for all histories~$h$, $\sigma_A(h) = \sigma_A(\last(h))$.
\end{mydefinition}

\begin{mydefinition}[Outcomes]
Let $C$ be a~coalition, and $\sigma_C$ a~strategy for~$C$. 
A~history~$h$ is \newdef{compatible} with the strategy~$\sigma_C$ if, for all~$k<\length h - 1$, $(\act_k(h))_A = \sigma_A(h_{\le k})$ for all~$A\in C$, and $\Tab(h_{k}, \act_k(h)) = h_{k+1}$.
A~play~$\rho$ is \newdef{compatible} with the strategy~$\sigma_C$ if all its prefixes are.
We~write~$\Out_{\calG}(\stat,\sigma_C)$ for the set of plays in~$\calG$ that
are compatible with strategy~$\sigma_C$ of~$C$ and have initial state~$\stat$,
these paths are called~\emph{outcomes} of $\sigma_C$ from~$\stat$.
We simply write~$\Out_{\calG}(\sigma_C)$ when $\stat = \stat_0$ and $\Out_{\calG}$ is the set of plays that are compatible with some strategy.
Note that when the coalition~$C$ is composed of all the players and the strategies are deterministic the outcome is unique.


An \newdef{objective}~$\Omega$ is a set of plays and a strategy $\sigma_C$ is said \newdef{winning} for objective $\Omega$ if all its outcomes belong to $\Omega$.

\end{mydefinition}


\begin{mydefinition}[Probability measure induced by a strategy profile]
Given a strategy profile $\sigma_\Agt$, the \newdef{conditional probability} of $a_\Agt$ given history~$h\in\histories$ is $\sigma_\Agt(a_\Agt \mid h) = \prod_{A \in \Agt} \sigma_A(a_A\mid h)$.
The probabilities $\sigma_\Agt(a_\Agt \mid h)$ induce a probability measure on the Borel $\sigma$-algebra over $(\Stat \cdot \Act^\Agt)^\omega$ as follows:
the probability of a basic open set $h \cdot (\Stat \cdot \Act^\Agt)^\omega$ equals the product $\prod_{j=1}^n \sigma_\Agt(h^\Act_{j,A} \mid h_{\le j})$ if $h_0 = s_0$ and $\Tab(h_j,h^\Act_{\Agt,j}) = h_{j+1}$ for all $1 \le j < n$; in all other cases, this probability is 0.
By Carathéodory's extension theorem, this extends to a unique probability measure assigning a probability to every Borel subset of $(\Stat \cdot \Act^\Agt)^\omega$, which we denote by $\Pr^{\sigma_\Agt}$.
We denote by $\E^{\sigma_\Agt}$ the expectation operator that corresponds to $\Pr^{\sigma_\Agt}$, that is $\E^{\sigma_\Agt}(f) = \int f\ \text{d}\Pr^{\sigma_\Agt}$ for all Borel measurable functions $f \colon (\Stat \cdot \Act^\Agt)^\omega \mapsto \mathbb{R} \cup \{ \pm \infty\}$. 
The expected payoff for player $A$ of a strategy profile $\sigma_\Agt$ is $\payoff_A(\sigma_\Agt) = \E^{\sigma_\Agt}(\MP_{w_A})$.
\end{mydefinition}


\subsection{Equilibria notions}

We now present the different solution concepts we will study.
Solution concepts are formal descriptions of ``good'' strategy profiles.
The most famous of them is Nash equilibrium~\cite{nash50}, in which no single player can improve the outcome for its own preference relation, by only changing its strategy.
This notion can be generalised to consider coalitions of players, it is then called a resilient strategy profile.
Nash equilibria correspond to the special case of $1$-resilient strategy profiles.

\begin{mydefinition}[Resilience~\cite{aumann1959acceptable}]
Given a coalition~$C \subseteq \Agt$, a strategy profile~$\sigma_{\Pl}$ is \newdef{$C$-resilient} if for all agents $A$ in $C$, $A$ cannot improve her payoff even if all agents in $C$ change their strategies, \ie $\sigma_{\Pl}$ is said $C$-resilient when:
\[
\forall \sigma'_C \in \Strat_\calG(C).\ \forall A\in C.\ 
\payoff_{A} (\sigma_{-C}, \sigma'_C) \le \payoff_A(\sigma_\Pl)
\]
Given an integer $k$, we say that a strategy profile is $k$-resilient if it is $C$-resilient for every coalition~$C$ of size $k$.
\end{mydefinition}

\begin{mydefinition}[Immunity~\cite{abraham2006distributed}]
  Immune strategies ensure that players \emph{not} deviating are not too much affected by deviation.
  Formally, a strategy profile~$\sigma_{\Agt}$ is \newdef{($C,r$)-immune} if all players not in $C$, are not worse off by more than $r$ if players in $C$ deviates, \ie when:
\[
\forall \sigma'_C \in \Strat_\calG(C).\ \forall A\in \Agt \setminus C.\ 
\payoff_A(\sigma_\Pl) - r \le  \payoff_A(\sigma_{-C}, \sigma'_C) 
\]
Given an integer $t$, a strategy profile is said \newdef{($t,r$)-immune} if it is ($C,r$)-immune for every coalition~$C$ of size $t$.
Note that $t$-immunity as defined in~\cite{abraham2006distributed} corresponds to $(t,0)$-immunity.
\end{mydefinition}

\begin{mydefinition}[Robust Equilibrium~\cite{abraham2006distributed}]
Combining resilience and immunity, gives the notion of robust equilibrium:
a strategy profile is a \newdef{$(k,t,r)$-robust equilibrium} if it is both $k$-resilient and $(t,r)$-immune.
\end{mydefinition}

The aim of this article is to characterise robust equilibria in order to construct the corresponding strategies, and precisely describe the complexity of the following decision problem for mean-payoff games.
\begin{mydefinition}[Robustness Decision Problem]
Given a game~$\calG$, integers~$k$, $t$, rational $r$
does there exist a profile $\sigma_\Agt$, that is a $(k,t,r)$-robust equilibrium~$\sigma_\Agt$ in $\calG$?
\end{mydefinition}

\subsection{Undecidability}\label{sec:undec}
We show that allowing strategies that can use both randomisation and memory leads to undecidability.
The problem of existence of Nash equilibria has been shown to be undecidable if we put constraints on the payoffs in~\cite{ummels2011complexity}.
The proof was improved to involve only 3 players and no constraint on the payoffs for games with terminal-reward (the weights are non-zero only in terminal vertices of the game)~\cite{BMS14}.
This corresponds to the particular case where $k=1$, $t=0$ for the robustness decision problem.
Therefore, the following is a corollary of \cite[Thm.13]{BMS14}.
\begin{theorem}
  For randomised strategies, the robustness problem is undecidable even for 3 players, $k=1$ and $t=0$.
\end{theorem}

To recover decidability, two restrictions are natural.
In the next section, we will show that for randomised strategies with no memory, the problem is decidable.
The rest of the article is devoted to the second restriction, which concerns pure strategies.


\section{Stationary strategies}\label{sec:randomised}


We use the \emph{existential theory of reals}, which is the set of all existential first-order sentences that hold in the ordered field $\mathfrak{R} := (\mathbb{R},+,\cdot,0,1,\le)$, to show that the robustness decision problem is decidable.
The associated decision problem is in the \PSPACE\ complexity class~\cite{Can88}.
However, since the system of equation we produce is of exponential size, we can only show that the robustness problem for (randomised) stationary strategies is in \EXPSPACE.

\mysubsubsection{Encoding of strategies}
We encode a stationary strategy $\sigma_A$, by a tuple of real variables~$(\varsigma_{\stat,\shortAct})_{\stat\in \Stat,\shortAct\in \Act}\in \mathbb{R}^{\Stat\times\Act}$ such that $\varsigma_{s,a} = \sigma_A(a \mid s)$ for all $s\in\Stat$ and $a\in\Act$.
We then write a formula saying that these variables describe a correct stationary strategy.
The following lemma is a direct consequence of the definition of strategy.
\begin{lemma}\label{lem:strategy-equation}
  Let $(\varsigma_{s,a})_{s\in S, a \in \Act}$ be a tuple of real variables.
  The mapping $\sigma \colon s,a \mapsto \varsigma_s$ is a stationary strategy if, and only if, $\varsigma$ is a solution of the following equation:
  \begin{center}
  \(
  \mu(\varsigma) := 
  \bigwedge_{s\in S} \left( \sum_{a \in \Act } \varsigma_{s,a} = 1 \land \bigwedge_{a\in \Act} \varsigma_{s,a} \ge 0 \right)
  \)
  \end{center}
\end{lemma}

\mysubsubsection{Payoff of a profile}
We now give an equation which links a stationary strategy profile and its payoff.
For this we notice that a concurrent game where the stationary strategy profile has been fixed corresponds to a \emph{Markov reward process}~\cite{Puterman94}.
We recall that a Markov reward process is a tuple $\langle S, P, r \rangle$ where:
\begin{inparaenum} \item $S$ is a set of states; \item $P \in \mathbb{R}^{S\times S}$ is a transition matrix; \item $r \colon S \mapsto \mathbb{R}$ is a reward function. \end{inparaenum}
The expected value is then the expectation of the average reward of $r$.
This value for each state is described by a system of equations in \cite[Theorem~8.2.6]{Puterman94}.
We reuse these equations to obtain Thm.~\ref{thm:payoff-stationary}.
Details of the proof are in the appendix.

\begin{collect}{appendix-randomised}{}{}
We will recall how to compute values of \newdef{Markov reward process} (Lem.~\ref{lem:solution-MRP}), make explicit the correspondence with our problem on concurrent games (Lem.~\ref{lem:CGS-to-MRP}), and then give equations describing the payoff of strategy profile.

\begin{definition}[Markov reward processes]
  A \newdef{Markov reward process} is a tuple $\langle S, P, r \rangle$ where:
  \begin{inparaenum} \item $S$ is a set of states; \item $P \in \mathbb{R}^{S\times S}$ is a transition matrix; \item $r \colon S \mapsto \mathbb{R}$ is a reward function. \end{inparaenum}
  The probability of a history $h \in S^\ast$, is $\prod_{0 \le i < |h|} P(h_i,h_{i+1})$.
  The expected value of a Markov reward process is then defined similarly to the expected value of a concurrent game, as the expectation of $\MP_r$.
  The \newdef{gain}~$g \in \mathbb{R}^S$ is a vector such that for each state~$s$, $g(s)$ is the expected value from state $s$.
\end{definition}
\end{collect}

\begin{collect}{appendix-randomised}{}{}
\subsection{Proof of Thm.~\ref{thm:payoff-stationary}}
We recall the result of \cite{Puterman94} that we will use.
\begin{theorem}[{\cite[Theorem~8.2.6]{Puterman94}}]\label{thm:MRP}
Let $\langle S,P,r \rangle$ be Markov reward process.
We consider real variables of the form $\gamma_s$ and $\beta_s$ for each $s\in S$.
The gain is uniquely characterised has the solution for $\gamma$ of the equations
\((I - P) \cdot \gamma = 0\) and \(\gamma + (I - P) \cdot \beta = r\) where $I$ is the identity matrix.
This solution is such that $\gamma$ is equal to the gain\footnote{$\beta$ is equal to a quantity called the bias~\cite{Puterman94}}.
\end{theorem}

Thanks to this result, we can show that solutions of the following equation are such that the mapping $g \colon s \mapsto \gamma_s$ corresponds to the gain:
\begin{equation}\label{eq:MRP}
  \exists \beta\in \mathbb{R}^S.\ \bigwedge_{s \in S} \left( \gamma_s = \sum_{s'\in S} P_{s,s'} \cdot \gamma_{s'}\right) \land 
  \bigwedge_{s \in S} \left( \gamma_s + \beta_s = r(s) + \sum_{s'\in S} P_{s,s'} \cdot \beta_{s'} \right)
\end{equation}

\begin{lemma}\label{lem:solution-MRP}
  In a Markov reward process~$\langle S,P,r\rangle$, the solution $\gamma$ of equation~\eqref{eq:MRP} is such that for each state $s$, $g(s) = \gamma_s$ where $g$ is the gain function.
\end{lemma}
\begin{proof}
  By using \ref{thm:MRP}, $g$ (the gain) and $b$ (the bias) are the solutions for $\gamma$ and $\beta$ of the equations of \((I - P) \cdot \gamma = 0\) and \(\gamma + (I - P) \cdot \beta = r\).
  We can rewrite the equation and since we are only interested in $g$ we can abstract $b$:
  \[
  \exists \beta \in \mathbb{R}^S.\ \bigwedge_{s \in S} \left( \gamma_s - \sum_{s'\in S} P_{s,s'} \cdot \gamma_{s'} = 0\right) \land 
  \bigwedge_{s \in S} \left( \gamma_s + \beta_s - \sum_{s'\in S} P_{s,s'} \cdot \beta_{s'} = r(s) \right) 
  \]
  Which we then rewrite in the form of equation~\eqref{eq:MRP}.
\end{proof}

We now make explicit the correspondence with concurrent games.
Let $\calG$ be a concurrent game, $\sigma_\Agt$ be a strategy profile and $w$ a weight function.
We define the Markov reward process $\MRP(\calG,\sigma_\Agt, w) = \langle \Stat, P , w \rangle$ where for all $(s,s')\in \Stat\times \Stat$, $P(s,s') = \sum_{a_\Agt \mid \Tab(s,a_\Agt) = s'} \prod_{A \in \Agt} \sigma_A(a_A \mid s)$.

\begin{lemma}\label{lem:CGS-to-MRP}
  The expectation for $\MP_w$
  of the strategy profile $\sigma_\Agt$ is equal to the expected value of the Markov reward process $\MRP(\calG,\sigma_\Agt,w)$.
\end{lemma}
\begin{proof}
  For proving the lemma, since the weight functions are the same, it is enough to prove that the probability of an history in the Markov reward process is the same as its probability in the concurrent game knowing that the strategy profile $\sigma_\Agt$.
  This is done by induction over the length of history $h$.
  The case of length 1 is obvious.
  Now assume that $h$ is a history that has same probability~$p(h)$ in $\calG$ knowing $\sigma_\Agt$ and in $\MRP(\calG,\sigma_\Agt,A)$.
  The probability of $h\cdot s$ in $\calG$ knowing $\sigma_\Agt$ is the sum over actions profiles~$a_\Agt$ that lead from $s$ to $s'$ of the probability that the strategy profile $\sigma_{\Agt}$ chooses $a_{\Agt}$.
  This equals $p(h) \cdot \sum_{a_\Agt \mid \Tab(s,a_\Agt) = s'} \prod_{A \in \Agt} \sigma_A(a_A \mid s)$, which is also equal to the probability of a transition from $s$ to $s'$ in $\MRP(\calG,\sigma_\Agt,w)$ times $p(h)$.
  Therefore the probability of $h \cdot s$ is equal in $\calG$ knowing $\sigma_\Agt$ and in $\MRP(\calG,\sigma_\Agt,w)$.
  This shows the property.
\end{proof}
\end{collect}

\begin{collect*}{appendix-randomised}{
\begin{theorem}\mylabel{Thm}{thm:payoff-stationary}
  Let $\sigma_\Agt$ be a stationary strategy profile.
  The expectation for $\MP_w$ of $\sigma_\Agt$ from $s$ is the component $\gamma_s$ of the solution $\gamma$ of the equation:\\
  $\psi(\gamma,\sigma_\Agt,w)$ $:= \exists \beta\in \mathbb{R}^S.\ $ $\bigwedge_{s \in S} \left( \gamma_s = \sum_{a_\Agt \in\Act^\Agt} \gamma_{\Tab(s,a_\Agt)} \cdot \prod_{A \in \Agt} \sigma_A(a_A\mid s) \right)$ \\
  $ \land 
  \bigwedge_{s \in S} \left( \gamma_s + \beta_s = w(s) + \sum_{a_\Agt \in \Act^\Agt} \beta_{\Tab(s,a_\Agt)} \cdot\prod_{A \in \Agt}  \sigma_A(a_A \mid s) \right)$
\end{theorem}
}{}{}{}
\end{collect*}
\begin{collect}{appendix-randomised}{}{}
\begin{proof}
  Thanks to Lem.~\ref{lem:CGS-to-MRP} the expected payoff is the same than the value of $\MRP(\calG,\sigma_\Agt,w)$.
  If we apply Lem.~\ref{lem:solution-MRP} to this Markov reward process, and replace $P$ by its expression given from $\sigma_\Agt$, $r$ by $w$, we obtain that its gain is given by the equation:
  \begin{align*}
  \exists \beta\in \mathbb{R}^S.\ & \bigwedge_{s \in S} \left( \gamma_s = \sum_{s'\in S}  \sum_{\{a_\Agt \mid \Tab(s,a_\Agt) = s'\}} \gamma_{s'} \cdot \prod_{A \in \Agt} \sigma_A(a_A \mid s) \right) \\
  & \land 
  \bigwedge_{s \in S} \left( \gamma_s + \beta_s = w(s) + \sum_{s'\in S}  \sum_{\{a_\Agt \mid \Tab(s,a_\Agt) = s'\}} \beta_{s'} \cdot \prod_{A \in \Agt} \sigma_A(a_A \mid s) \right)
  \end{align*}
  We can then rewrite this, noticing that $\sum_{s'\in S} \sum_{\{a_\Agt \mid \Tab(s,a_\Agt) = s'\}} f(s,s',a_\Agt)$ is the same as \linebreak[3] $\sum_{a_\Agt \in \Act^\Agt} f(s,\Tab(s,a_\Agt),a_\Agt)$ and we obtain the desired expression.
\end{proof}
\end{collect}

\subsubsection{Optimal payoff of a deviation}

We now want to keep the strategy profile fixed but also allow deviations and investigate what is the maximum payoff that a coalition can achieve by deviating.
The system that is obtained is then a \emph{Markov decision processes}.
We recall that a Markov decision process~(MDP) is a tuple $\langle S,A,P, r\rangle$, where:
\begin{inparaenum}
\item $S$ is a the non-empty, countable set of \newdef{states}.
\item $A$ is a set of \newdef{actions}.
\item $P \colon S\times A \times S \mapsto [0,1]$ is the \newdef{transition relation}.
  It is such that for each $s \in S \setminus S_\shortEve$ and $a \in A$, $\sum_{(s,a,s') \in S \times A \times S} P(s,a,s') = 1$.
\item $r \colon S \mapsto \mathbb{R}$ is the \newdef{reward function}.
\end{inparaenum}
The \emph{optimal value} is then maximal expectation that can be obtained by a strategy.
The optimal values in such systems can be described by a linear program~\cite[Section~9.3]{Puterman94}.
We reuse this linear program to characterise optimal values against a strategy profile $C$.
Details of the proof can be found in the appendix.

\begin{collect}{appendix-randomised}{}{}
We will recall how to compute values of such processes (Lem.~\ref{lem:solution-MDP}), draw the link with our problem(Lem.~\ref{lem:CGS-to-MDP}), and then give equations linking the maximum payoff a coalition can achieve with the given strategy profile (Thm.~\ref{thm:value-stationary}).

\begin{definition}[\cite{Puterman94,ummels2011complexity}]
  A \newdef{Markov decision process}~(MDP) is given by a tuple $\langle S,A,P, r\rangle$, where:
  \begin{itemize}
  \item $S$ is a the non-empty, countable set of \newdef{states}.
  \item $A$ is a set of \newdef{actions}.
  \item $P \colon S\times A \times S \mapsto [0,1]$ is the \newdef{transition relation}.
    It is such that for each $s \in S \setminus S_\shortEve$ and $a \in A$, $\sum_{(s,a,s') \in S \times A \times S} P(s,a,s') = 1$.
  \item $r \colon S \mapsto \mathbb{R}$ is the \newdef{reward function}.
  \end{itemize}
  Given a policy $\sigma \colon S \to A$, the probability of a history $h \in S^\ast$, is $\prod_{0 \le i < |h|} \sum_{a\in A} P(h_i,\sigma(a \mid h_i),h_{i+1})$.
  The \newdef{expected value}~$E^\sigma(r)$ of a policy~$\sigma$ in a Markov reward process is then defined similarly to the expected value of a concurrent game.
  The \newdef{optimal average reward} is the maximum over the policies of the expected value: 
  the optimal average reward of the MDP~$M$ from state $s$ is written $v(M,s)$ and equals $\sup_{\{\sigma \colon S \to A\}} E^\sigma(\MP_r)$.
\end{definition}

In \cite[Section~9.3]{Puterman94} it is shown that the optimal average reward of a MDP is characterised by the following linear programs:
\begin{quote}
  Minimise $\sum_{s\in S} \alpha_s \cdot \gamma_s$ where for all states~$s$, $\alpha_s > 0$ and $\sum_{s\in S} \alpha_s = 1$ subject to:
  \begin{eqnarray}\label{eq:MDP}
    \bigwedge_{a\in \Act} \bigwedge_{s\in S} \gamma_s & \ge \sum_{s' \in S} P(s,a,s') \cdot \gamma_{s'} \nonumber \\
    \land \bigwedge_{a\in \Act} \bigwedge_{s\in S} \gamma_s & \ge r(s) + \sum_{s' \in S} P(s,a,s') \cdot \beta_{s'} - \beta_{s}
  \end{eqnarray}
\end{quote}

Thanks to this linear program we can characterise optimal values against a strategy profile $C$.
\end{collect}

\begin{collect}{appendix-randomised}{}{}
We reformulate the results of \cite[Section~9.3]{Puterman94} in the following lemma.
\begin{lemma}\label{lem:solution-MDP}
  Equation~\eqref{eq:MDP} is fulfilled by $\gamma = (v(M,s))_{s \in S}$ and if it is fulfilled by some vector $\gamma$ then $\gamma_s \ge v(M,s)$ for all states $s$.
\end{lemma}
\begin{proof}
  That equation~\eqref{eq:MDP} if fulfilled by $v(M,s)$ is obvious because it is the solution of the linear program.
  Assume now that some $\gamma$ satisfies $\phi$.

  First notice that $\sum_{s \in S} \gamma_s - v(M,s) \ge 0$ because $\gamma$ is a solution of equation~\ref{eq:MDP} and $v(M,s)$ is the solution that minimises $\sum_{s\in S} \frac{\gamma_s}{|S|}$ (taking $\alpha_s = \frac{1}{|S|}$ for all $s\in S$).

  Towards a contradiction assume there is a state~$t$, such that $\gamma_t < v(M,t)$.
  Given some $\varepsilon$ with $0 < \varepsilon < 1$, we let $\alpha_t = 1 - \varepsilon$ and $\alpha_s = \frac{\varepsilon}{|S|-1}$ for each $s\ne t$.
  We have that the constraints $\alpha_s > 0$ for all states $s$ and $\sum_{s\in S} \alpha_s = 1$ are satisfied.
  Moreover:
  \begin{align*}
    \sum \alpha_s & \cdot \gamma_s - \sum \alpha_s \cdot v(M,s) \\
    & = \sum \alpha_s \cdot (\gamma_s - v(M,s)) \\
    & = \alpha_t \cdot (\gamma_t - v(M,t)) + \sum_{s \ne t} \alpha_s \cdot (\gamma_s - v(M,s)) \\
    & = (\gamma_t - v(M,t)) - \varepsilon \cdot (\gamma_t-v(M,t)) + \sum_{s \ne t} \varepsilon \cdot \frac{(\gamma_s - v(M,s))}{|S|-1} \\
    & = (\gamma_t - v(M,t)) - \varepsilon \cdot \left(\left(1 + \frac{1}{|S|-1}\right)\cdot (\gamma_t-v(M,t)) - \sum_{s \in S} \frac{\gamma_s - v(M,s)}{|S|-1} \right)\\
  \end{align*}
  We write $\delta = \frac{\gamma_t - v(M,t)}{\left(1 + \frac{1}{|S|-1}\right)\cdot (\gamma_t-v(M,t)) - \sum_{s \in S} \frac{\gamma_s - v(M,s)}{|S|-1}}$, so that:
  \[ \sum \alpha_s \cdot \gamma_s - \sum \alpha_s \cdot v(M,s) = (\gamma_t - v(M,t)) \left(1 - \frac{\varepsilon}{\delta} \right)\]

  We have that $\delta$ is greater than $0$ because $\gamma_t - v(M,t) < 0$ by hypothesis, and we showed that $\sum_{s \in S} \gamma_s - v(M,s) \ge 0$ so $\delta > 0$.
  We can then consider some $\varepsilon$ such that $0 < \varepsilon < \min \{ 1 , \delta \}$.
  We obtain that 
  \[ \sum \alpha_s \cdot \gamma_s - \sum \alpha_s \cdot v(M,s) < (\gamma_t - v(M,t)) \]
  Which means that $\sum \alpha_s \cdot \gamma_s < \sum \alpha_s \cdot v(M,s)$ and contradicts that $v(M,s)$ is the solution of $\phi$ that minimises $\sum \alpha_s \cdot \gamma_s$.

  We conclude that if vector $\gamma$ satisfies equation~\eqref{eq:MDP} then $\gamma_s \ge v(M,s)$ for all states $s$.
\end{proof}

Given a coalition $C$, a stationary strategy profile~$\sigma_C$ for this coalition, and a weight function~$w$,
we define a Markov decision process $\MDP(\calG,\sigma_{C},w)$ which represents possible executions of the game when strategies for the coalition $C$ have been fixed.
We define $\MDP(\calG,\sigma_{C},w) = \langle \Stat, \Act^{\Agt\setminus C}, P, w \rangle$ where $P$ is such that the probability of a transition from $s$ to $s'$ with action $a_{-C}$ is:
\[P(s,a_{-C},s') = \sum_{a_C \mid \delta(s,a_C,a_{-C}) = s'} \prod_{A\in C} \sigma_A(a_A \mid h).\]
Given a strategy $\sigma_{-C}$ of the coalition $\Agt \setminus C$ in $\calG$, we define its projection~$\pi(\sigma_{-C})$ in $\calG(\sigma_{-C})$: 
for all histories $h$, $\pi(\sigma_{-C})(a_{-C} \mid h) = \prod_{A \in \Agt \setminus C} \sigma_A(a_A \mid h)$. 

\begin{lemma}\label{lem:CGS-to-MDP}
  The expectation for $\MP_w$ of profile $(\sigma_C,\sigma_{-C})$ in the game~$\calG$ is the same as the expected value of the policy $\pi(\sigma_{-C})$ in $\MDP(\calG,\sigma_C,w)$.
\end{lemma}
\begin{proof}
  Since the weights are given by the same function, it is enough to prove that the probability of an history in $\calG$ knowing $(\sigma_C,\sigma_{-C})$ is the same as the probability in $\MDP(\calG,\sigma_C,w)$ knowing $\sigma_{-C}$.
  We do this by induction. 
  This is obvious for a history~$h$ of length~$1$.
  We now assume that the probability~$p(h)$ of some history~$h$ is the same in $\calG$ knowing $(\sigma_C,\sigma_{-C})$ and in $\MDP(\calG,\sigma_C,w)$ knowing $\sigma_{-C}$.
  Let $s$ be a state, the probability of $h \cdot s$ in $\calG$ knowing $(\sigma_C,\sigma_{-C})$ is:\linebreak[3]
  $p(h) \cdot \sum_{a_\Agt \mid \delta(\last(h),a_\Agt) = s} \prod_{A \in \Agt} \sigma_A(a_A \mid h)$.
  This can be rewritten as: \linebreak[3]
  $p(h) \cdot \sum_{a_\Agt \mid \delta(\last(h),a_\Agt) = s} \pi(\sigma_{-C})(a_{-C} \mid h) \prod_{A\in C} \sigma_A(a_A \mid h)$
  which is equal to the probability of $h\cdot s$ in $\MDP(\calG,\sigma_C,w)$ knowing $\sigma_{-C}$.
  This shows that the expectation for $w$ of profile $(\sigma_C,\sigma_{-C})$ in the game~$\calG$ is the same as the expected payoff for $A$ of the strategy $\pi(\sigma_{-C})$ in the $\MDP(\calG,\sigma_C,w)$.
\end{proof}
\end{collect}

\begin{collect*}{appendix-randomised}{
\begin{theorem}\label{thm:value-stationary}
  Let $\sigma_\Agt$ be a stationary strategy profile, $C$ a coalition, $s$ a state and $w$ a weight function.
  The highest expectation $\Agt \setminus C$ can obtain for $w$ in $\calG$ from $s$: $\sup_{\sigma_{-C}} E^{\sigma_C, \sigma_{-C}}(\MP_w,s)$, 
  is the smallest~$\gamma_s$ component of a solution of the system of inequation:\\
  $\phi(\gamma,\sigma_C,w) := $
  $\bigwedge_{a_{-C}\in \Act^{\Agt \setminus C}} \bigwedge_{s\in S} \gamma_s$ $ \ge \sum_{a_C \in \Act^C}  \gamma_{\delta(s,a_C,a_{-C})} \cdot \prod_{A\in C} \sigma_A(a_A \mid s)$
  $\land \bigwedge_{a_{-C} \in \Act^{\Agt \setminus C}} \bigwedge_{s\in S} \gamma_s$ $\ge w(s) - \beta_{s} + \sum_{a_C \in \Act^C}  \left(\beta_{\delta(s,a_C,a_{-C})} \cdot \prod_{A\in C} \sigma_A(a_A \mid s)\right)$
\end{theorem}
}{}{}{}
\end{collect*}

\begin{collect}{appendix-randomised}{}{}
\begin{proof}
  We proved in Lem.~\ref{lem:CGS-to-MDP} that the expected payoff of a profile in $\calG$ is the same as the expected payoff of its projection in $\MDP(\calG,\sigma_C,w)$.
  Replacing the transition relation in the equation of $\phi$ by its expression in $\MDP(\calG,\sigma_C,w)$, and $r$ by $w$, we obtain by Lem.~\ref{lem:solution-MDP} that the expected payoff of the profile is the minimal solution of:
  \begin{align*}
    \bigwedge_{a_{-C}\in \Act^{-C}} \bigwedge_{s\in S} \gamma_s & \ge \sum_{s' \in S}  \sum_{\{a_C \mid \delta(s,a_C,a_{-C}) = s'\}}  \gamma_{s'} \cdot \prod_{A\in C} \sigma_A(a_A \mid s)    \\
    \land \bigwedge_{a_{-C}\in \Act^{-C}} \bigwedge_{s\in S} \gamma_s & \ge w(s) - \beta_{s} + \sum_{s' \in S}   \sum_{\{a_C \mid \delta(s,a_C,a_{-C}) = s'\}}  \left(\beta_{s'} \cdot \prod_{A\in C} \sigma_A(a_A \mid s) \right)
  \end{align*}
  We notice that $\sum_{s'\in S} \sum_{\{a_C \mid \delta(s,a_C,a_{-C}) = s'\}} f(s,s',a_C)$ is the same as:\linebreak[3]
  $\sum_{a_C \in \Act^C} f(s,\delta(s,a_C,a_{-C}),a_C)$, and rewrite the equation:
  \begin{align*}
    \bigwedge_{a_{-C}\in \Act^{-C}} \bigwedge_{s\in S} \gamma_s & \ge \sum_{a_C \in \Act^C}  \gamma_{\delta(s,a_C,a_{-C})} \cdot \prod_{A\in C} \sigma_A(a_A \mid s)    \\
    \land \bigwedge_{a_{-C} \in \Act^{-C}} \bigwedge_{s\in S} \gamma_s & \ge w(s) - \beta_{s} + \sum_{a_C \in \Act^C}  \left(\beta_{\delta(s,a_C,a_{-C})} \cdot \prod_{A\in C} \sigma_A(a_A \mid s)\right)
  \end{align*}
  Therefore $\sup_{\sigma_{-C}} E^{\sigma_C, \sigma_{-C}}(\MP_w,s)$ equals the minimal $\gamma_s$ that is part of a solution of this equation.
\end{proof}
\end{collect}

\subsubsection{Expressing the existence of robust equilibria}

Putting these results together, we obtain the results. 
Intuitively, $\psi(\gamma,\sigma_\Agt,w_A)$ ensures that $\gamma$ correspond to the payoff for $\sigma_\Agt$ of $A$, and $\phi(\gamma',\sigma_{-C},w_A)$ makes $\gamma'$ correspond to the payoff for a deviation of $\sigma_C$.

\begin{theorem}\label{thm:summary-randomized}
  Let~$\sigma_\Agt$ be a strategy profile.
  \begin{itemize}
  \item
    It is $C$-resilient if, and only if, it satisfies the formula:
    \begin{center}
  \(
  \rho(C,\sigma_\Agt) :=
  \bigwedge_{A\in C} \exists \gamma,\gamma'.\ \left(\psi(\gamma,\sigma_\Agt,w_A) \land \phi(\gamma',\sigma_{-C},w_A) \land (\gamma' \le \gamma)\right) \)
  \end{center}
  \item It is $C,r$-immune if, and only if, it satisfies equation:
    \begin{center}
  \(
  \iota(C,\sigma_\Agt) :=
  \bigwedge_{A\not\in C} \exists \gamma,\gamma'.\ \left(\psi(\gamma,\sigma_\Agt,A) \land \phi(\gamma',\sigma_{-C},-w_A) \land \gamma - r \le -\gamma'\right) \)
    \end{center}
  \item
  There is a robust equilibria if, and only if, the following equation is satisfiable:
  \begin{center}
  \(
  \exists \varsigma \in \mathbb{R}^{\Agt \times \Stat\times \Act}.\ \mu(\varsigma) \land 
  \bigwedge_{C \subseteq \Agt \mid |C| \le k} \rho(C,\varsigma)
  \land 
  \bigwedge_{C \subseteq \Agt \mid |C| \le t} \iota(C,\varsigma)
  \)
    \end{center}
  \end{itemize}
\end{theorem}

\begin{collect}{appendix-randomised}{}{}
\subsection{Proof of Thm.~\ref{thm:summary-randomized}}
\begin{lemma}\label{lem:resilience-equation}
  The strategy profile~$\sigma_\Agt$ is $C$-resilient if, and only if, it satisfies the formula:
  \[  
  \rho(C,\sigma_\Agt) :=
  \bigwedge_{A\in C} \exists \gamma,\gamma'.\ \left(\psi(\gamma,\sigma_\Agt,w_A) \land \phi(\gamma',\sigma_{-C},w_A) \land (\gamma' \le \gamma)\right) \]
\end{lemma}
\begin{proof}
  \fbox{$\Rightarrow$}
  Assume $\sigma_\Agt$ is $C$-resilient, and let $A\in C$.
  Consider $\gamma$ the payoff of $A$ in $\Out_\calG(\sigma_\Agt)$, and $\gamma'$ the highest payoff $C$ can obtain for $A$ against $\sigma_{-C}$: $\gamma' = \sup_{\sigma_{-C}} E^{\sigma_C, \sigma_{-C}}(\MP_{w_A},s)$.
  By Thm.~\ref{thm:payoff-stationary}, $\gamma$ makes $\psi(\gamma,\sigma_\Agt,w_A)$ hold, and by Thm.~\ref{thm:value-stationary}, $\gamma'$ makes $\phi(\gamma',\sigma_{-C},w_A)$ hold.
  Now since $\sigma_\Agt$ is resilient, $C$ cannot improve the payoff of $A$ and therefore $\gamma' \le \gamma$.
  Hence $\left(\psi(\gamma,\sigma_\Agt,w_A) \land \phi(\gamma',\sigma_{-C},w_A) \land (\gamma' \le \gamma)\right)$ is satisfied by this choice of $\gamma$ and $\gamma'$.
  As we can do the same for all $A\in C$, $\rho(C,\sigma_\Agt)$ holds.

  \fbox{$\Leftarrow$}
  The requirement that $\psi(\gamma,\sigma_\Agt,w_A)$ holds, ensures that the payoff for $A$ of $\sigma_\Agt$ is $\gamma$ (see Thm.~\ref{thm:payoff-stationary}).
  The requirement $\phi(\gamma',\sigma_{-C},w_A)$, ensures that the maximal payoff for $A$ that coalition $C$ can obtain against $\sigma_{-C}$ is less than $\gamma'$ (see Thm.~\ref{thm:value-stationary}).
  Finally the requirement $\gamma' \le \gamma$, ensures that $\gamma'$ and therefore the best value $C$ can obtain, is smaller than $\gamma$, hence the coalition $C$ cannot improved the payoff of $A$ by deviating.
  This implies the resilience of $\sigma_\Agt$.
\end{proof}

\begin{lemma}\label{lem:immunity-equation}
  The strategy profile~$\sigma_\Agt$ is $C,r$-immune if, and only if, it satisfies equation:
  \[ 
  \iota(C,\sigma_\Agt) :=
  \bigwedge_{A\not\in C} \exists \gamma,\gamma'.\ \left(\psi(\gamma,\sigma_\Agt,A) \land \phi(\gamma',\sigma_{-C},-w_A) \land \gamma - r \le -\gamma'\right) \]
\end{lemma}
\begin{proof}
  \fbox{$\Rightarrow$}
  Assume $\sigma_\Agt$ is $C,r$-immune, and let $A\not\in C$.
  Consider $\gamma$ the payoff of $A$ in $\Out_\calG(\sigma_\Agt)$.
  By Thm.~\ref{thm:payoff-stationary}, $\psi(\gamma,\sigma_\Agt,A)$ holds.
  Consider also $\gamma'$ the negation of lowest payoff $C$ can obtain for $A$ against $\sigma_{-C}$: $\gamma' = - \inf_{\sigma_{-C}} E^{\sigma_C, \sigma_{-C}}(\MP_{w_A},s)$.
  We use the fact that the optimal values for limit inferior and superior coincide in MDPs~(see for instance~\cite[Thm.~9.1.3]{Puterman94} to get the following
  \begin{align*}
    - \inf_{\sigma_{-C}} E^{\sigma_C, \sigma_{-C}}(\MP_{w_A},s) &= \sup_{\sigma_{-C}} - E^{\sigma_C, \sigma_{-C}}(\MP_{w_A},s) \\
    &= \sup_{\sigma_{-C}} E^{\sigma_C, \sigma_{-C}}(-\MP_{w_A},s) \\
    &= \sup_{\sigma_{-C}} E^{\sigma_C, \sigma_{-C}}\left(\rho \mapsto \limsup_{n \to \infty} \frac{1}{n} \sum_{1\le k\le n} -w_A(\rho_k),s\right) \\
    &= \sup_{\sigma_{-C}} E^{\sigma_C, \sigma_{-C}}(\MP_{-w_A},s) & \text{(consequence of \cite{Puterman94})}\\
  \end{align*}
  By Thm.~\ref{thm:value-stationary}, $\gamma'$ makes $\phi(\gamma',\sigma_{-C},-w_A)$ hold.

  Now since $\sigma_\Agt$ is $C,r$-immune, the deviation of $C$ cannot make the payoff of $A\not\in C$ decrease by more than $r$.
  Therefore $\inf_{\sigma_{-C}} E^{\sigma_C, \sigma_{-C}}(\MP_{w_A},s) \ge \gamma -r$ and $-\gamma' \ge \gamma - r$.

  \fbox{$\Leftarrow$}
  The requirement that $\psi(\gamma,\sigma_\Agt,A)$ holds, ensures that the payoff for $A$ of $\sigma_\Agt$ is $\gamma$ (see Thm.~\ref{thm:payoff-stationary}).
  The requirement $\phi(\gamma',\sigma_{-C},-w_A)$, ensures that $\gamma' \ge \sup_{\sigma_{-C}} E^{\sigma_C, \sigma_{-C}}(\MP_{-w_A},s)$ (see Thm.~\ref{thm:value-stationary}).
  As we so in the proof of implication, $\sup_{\sigma_{-C}} E^{\sigma_C, \sigma_{-C}}(\MP_{-w_A},s) = - \inf_{\sigma_{-C}} E^{\sigma_C, \sigma_{-C}}(\MP_{w_A},s)$, thus $-\gamma' \le \inf_{\sigma_{-C}} E^{\sigma_C, \sigma_{-C}}(\MP_{w_A},s)$.
  Therefore coalition $C$ against $\sigma_{-C}$ cannot make the payoff of $A$ lower than $-\gamma'$, so with the additional constraint that $\gamma - r \le -\gamma'$, the payoff of $A$ cannot be decreased by more than $r$ by a deviation of $C$.
  This implies the immunity of $\sigma_\Agt$.
\end{proof}

\begin{theorem}\label{thm:robust-equation}
  There is a robust equilibria if, and only if, the following equation is satisfiable:
  \[ 
  \exists \tau \in \mathbb{R}^{\Agt \times \Stat\times \Act}.\ \mu(\tau) \land 
  \bigwedge_{C \subseteq \Agt \mid |C| \le k} \rho(C,\tau)
  \land 
  \bigwedge_{C \subseteq \Agt \mid |C| \le t} \iota(C,\tau)
  \]
\end{theorem}
\begin{proof}
  The formula $\mu(\tau)$ that for each $A\in\Agt$ the mapping $s,a \mapsto \tau_{A,s,a}$ described by $\tau$ corresponds to a stationary strategy (see Lem.~\ref{lem:strategy-equation}).
  Formula $\rho(C,\tau)$, says that for each coalition $C$ of size smaller than $k$, the profile described by $\tau$ is $C$-resilient (see Lem.~\ref{lem:resilience-equation}), hence it is $k$-resilient.
  Finally $\iota(C,\tau)$, says that for each coalition $C$ of size smaller than $t$, the profile described by $\tau$ is $(C,r)$-immune (see Lem.~\ref{lem:immunity-equation}), hence it is $(t,r)$-immune.
  This is therefore equivalent to the $(k,t,r)$-robustness of the strategy profile described by $\tau$.
\end{proof}
\end{collect}

\begin{theorem}
  The robustness problem is in \EXPSPACE\ for stationary strategies.
\end{theorem}
\begin{proof}
  By Thm.~\ref{thm:summary-randomized}, the existence of a robust equilibria is equivalent to the satisfiability of a formula in the existential theory of reals.
  This formula can be of exponential size with respect to $k$ and $t$, since a conjunction over coalitions of these size is considered.
  The best known upper bound for the theory of the reals in \PSPACE~\cite{Can88}, which gives the \EXPSPACE~upper bound for our problem.
\end{proof}



\section{Deviator Game}\label{sec:suspect}
We now turn to the case of non-randomised strategies.
In order to obtain simple algorithms for the robustness problem, we use a correspondence with zero-sum two-players game.
Winning strategies has been well studied in computer science and we can make use of existing algorithms.
We present the deviator game, which is a transformation of multiplayer game into a turn-based zero-sum game, such that there are strong links between robust equilibria in the first one and winning strategies in the second one.
This is formalised in Thm.~\ref{thm:dev-correct}.
Note that the proofs of this section are independent from the type of objectives we consider, and the result could be extended beyond mean-payoff objectives.

\begin{mydefinition}[Deviator]
The basic notion we use to solve the robustness problem is that of deviators.
It identifies players that cause the current deviation from the expected outcome.
  A \emph{deviator} from move $\shortAct_\Agt$ to $\shortAct'_\Agt$ is a player $D \in \Agt$ such that $\shortAct_D \ne \shortAct'_D$.
  We write this set of deviators:
  \(
  \dev(\shortAct_\Agt,\shortAct'_\Agt) = \{ A\in \Agt \mid \shortAct_A \ne \shortAct'_A \}.
  \)
  We extend the definition to histories
  and strategies by taking the union of deviator sets, formally $\dev(h,\sigma_\Agt) = \bigcup_{0 \le i < |h|} \dev(\act_i(h), \sigma_\Agt(h_{\le i}))$.
  It naturally extends to plays: if $\rho$ is a play, then $\dev(\rho,\sigma_\Agt) = \bigcup_{i\in\N} \dev(\act_i(\rho), \sigma_\Agt(\rho_{\le i}))$.
\end{mydefinition}

Intuitively, given an play~$\rho$ and a strategy profile~$\sigma_\Agt$, deviators represent the agents that need to change their strategies from $\sigma_\Agt$ in order to obtain the play $\rho$.
The intuition is formalised in the following lemma.

\begin{collect*}{appendix-suspect}{
\begin{lemma}\mylabel{Lem}{lem:deviator-path}
  Let $\rho$ be a play, $\sigma_\Agt$ a strategy profile and $C \subseteq \Agt$ a coalition.
  Coalition $C$ contains $\dev(\rho,\sigma_\Agt)$ if, and only if, there exists $\sigma'_C$ such that $\rho \in \Out_\calG(\rho_0, \sigma'_C, \sigma_{-C})$.
\end{lemma}
}{}{}{
\begin{proof}
  \fbox{$\Rightarrow$}
  Let $\rho$ be a play and $C$ a coalition which contains $\dev(\rho,\sigma_\Agt)$.
  We define $\sigma'_C$ to be such that for all $i$, $\sigma'_C(\rho_{\le i}) = (\act_i(\rho))_C$.
  We have that for all indices~$i$, $\dev(\rho_{i+1}, \sigma_\Agt(\act_i(\rho))) \subseteq C$.
  Therefore for all agents $A \not\in C$, $\sigma_A(\rho_{\le i}) = (\act_i(\rho))_A$.
  Then $\Tab(\rho_i, \sigma'_C(\rho_{\le i}), \sigma_{-C}(\rho_{\le i})) = \rho_{i+1}$.
  Hence $\rho$ is the outcome of the profile $(\sigma_{-C}, \sigma'_C)$. 

  \medskip

  \fbox{$\Leftarrow$} 
  Let $\sigma_\Agt$ be a strategy profile, $\sigma'_C$ a strategy for coalition $C$, and $\rho \in \Out_\calG(\rho_0,\sigma_{-C},\sigma'_C)$.
  We have for all indices~$i$ that $\act_i(\rho) = (\sigma_{-C}(\rho_{\le i}),\sigma'_C(\rho_{\le i}))$.
  Therefore for all agents~$A\not\in C$, $(\act_i(\rho))_A = \sigma_A(\rho_{\le i})$.
  Then $\dev(\act_i(\rho), \sigma_\Agt(\rho_{\le i})) \subseteq C$.
  Hence $\dev(\rho,\sigma_\Agt)\subseteq C$.
\end{proof}
}\end{collect*}

\subsection{Deviator Game}

We now use the notion of deviators to draw a link between multiplayer games and a two-player game that we will use to solve the robustness problem.
Given a concurrent game structure $\calG$, we define the deviator game~$\devgame$ between two players called \Eve and \Adam. 
Intuitively \Eve needs to play according to an equilibrium, while \Adam tries to find a deviation of a coalition which will profit one of its player or harm one of the others.
The states are in $\Stat' = \Stat \times 2^\Agt$; the second component records the deviators of the current history.
The game starts in $(s_0,\varnothing)$ and then proceeds as follows: from a state~$(s,D)$, \Eve chooses an action profile $\shortAct_\Agt$ and \Adam chooses another one $\shortAct'_\Agt$, then the next state is $(\Tab(s,\shortAct'_\Agt),D \cup \dev(\shortAct_\Agt,\shortAct'_\Agt))$.
In other words, \Adam chooses the move that will apply, but this can be at the price of adding players to the $D$ component when he does not follow the choice of \Eve.
The weights of a state $(s,D)$ in this game are the same than that of $s$ in $\calG$.
The construction of the deviator arena is illustrated in \figurename~\ref{fig:deviator-game}.

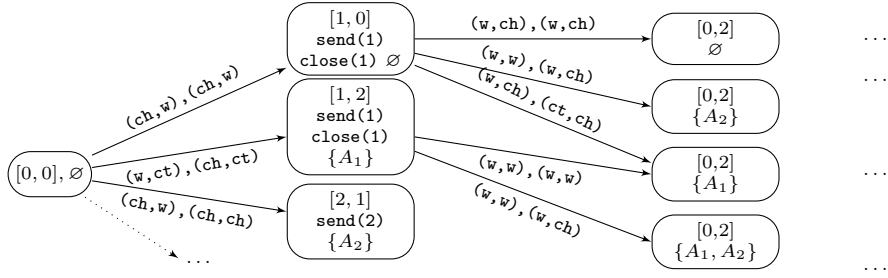
\begin{figure}[tb]
  \centering{\scriptsize
  \begin{tikzpicture}[xscale=1,yscale=0.9]
    \everymath{\scriptsize}
    \draw (0,0) node[draw,rounded corners=3mm,minimum size=6mm] (B) {$[0,0],\varnothing$};

    \draw (4,2) node[draw,text width=1.5cm,text centered,rounded corners=3mm] (BA1) {$[1,0]$ \texttt{send(1)} \texttt{close(1)} $\varnothing$};
    \draw (4,-0.7) node[draw,text width=1.5cm,text centered,rounded corners=3mm] (BA2) {$[2,1]$ \texttt{send(2)} $\{ A_2 \}$};
    \draw (4,0.7) node[draw,text width=1.5cm,text centered,rounded corners=3mm] (BA3) {$[1,2]$ \texttt{send(1)} \texttt{close(1)} $\{A_1\}$};
    \draw (2,-1.3) node(BA4) {\dots};

    \draw (8,2) node[draw,right,text width=1.5cm,text centered,rounded corners=3mm] (F) {[0,2] \\$\varnothing$};
    \draw (8,1) node[draw,right,text width=1.5cm,text centered,rounded corners=3mm,minimum size=6mm] (G) {[0,2] \\ $\{A_2\}$};
    \draw (8,0) node[draw,right,text width=1.5cm,text centered,rounded corners=3mm,minimum size=6mm] (H) {[0,2] \\ $\{A_1\}$};
    \draw (8,-1) node[draw,right,text width=1.5cm,text centered,rounded corners=3mm,minimum size=6mm] (I) {[0,2] $\{A_1,A_2\}$};

    \draw[-latex'] (B) -- node[above,sloped] {\texttt{(ch,w),(ch,w)}} (BA1);
    \draw[-latex'] (B) -- node[below,sloped] {\texttt{(ch,w),(ch,ch)}} (BA2);
    \draw[-latex'] (B) -- node[below,sloped] {\texttt{(w,ct),(ch,ct)}} (BA3);
    \draw[-latex',dotted] (B) -- (BA4.170);

    \draw[-latex'] (BA1) -- node[above,sloped] {\texttt{(w,ch),(w,ch)}} (F.180);
    \draw[-latex'] (BA1) -- node[above,sloped] {\texttt{(w,w),(w,ch)}} (G.180);
    \draw[-latex'] (BA1) -- node[above,sloped] {\texttt{(w,ch),(ct,ch)}} (H.170);
    \draw[-latex'] (BA3) -- node[below,sloped] {\texttt{(w,w),(w,w)}} (H.180);
    \draw[-latex'] (BA3) -- node[below,sloped] {\texttt{(w,w),(w,ch)}} (I.180);

    \draw (11,2) node {\dots};
    \draw (11,1.4) node {\dots};
    \draw (11,0) node {\dots};
    \draw (11,-1.4) node {\dots};
  \end{tikzpicture}
  }
  \caption{Part of the deviator game construction for the game of \figurename~\ref{fig:game-from-program}.
    Labels on the edges correspond to the action of \eve and the action of \adam.
    Labels inside the states are the state of the original game and the deviator component.
  }
  \label{fig:deviator-game}
\end{figure}

We now define some transformations between the different objects used in games $\calG$ and $\calD(\calG)$.
We define projections $\projun$,~$\projdeux$ and~$\projact$ from $\Stat'$ to $\Stat$, from $\Stat'$ to $2^\Agt$ and from 
$\Act^\Agt \times \Act^\Agt$ to $\Act^\Agt$ respectively.
They are given by $\projun(s,D) = s$, $\projdeux(s,D)=D$ and $\projact(\shortAct_\Agt,\shortAct'_\Agt) = \shortAct'_\Agt$. 
We~extend these projections to plays in a natural~way, letting
$\projpath(\rho) = \projun(\rho_0) \cdot \projact(\act_0(\rho)) \cdot \projun(\rho_1) \cdot \projact(\act_1(\rho)) \cdots$ and $\projdeux(\rho) = \projdeux(\rho_0) \cdot \projdeux(\rho_1)  \cdots$.
Note that for any play~$\rho$, and any index~$i$, $\projdeux(\rho_i) \subseteq \projdeux(\rho_{i+1})$, therefore $\projdeux(\rho)$ seen as a sequence of sets of coalitions is increasing and bounded by $\Agt$, its limit~$\devlimit(\rho)=\cup_{i\in\mathbb{N}} \projdeux(\rho_i)$ is well defined. 
Moreover to a strategy profile $\sigma_\Pl$ in $\calG$, we can naturally associate a strategy~$\projstrat(\sigma_\Pl)$ for \Eve in $\devgame$ such that for all histories $h$ by $\projstrat(\sigma_\Pl)(h) = \sigma_\Agt(\projpath(h))$.



The following lemma states the correctness of the construction of the deviator game, in the sense that it records the set of deviators in the strategy profile suggested by \adam with respect to the strategy profile suggested by \eve.
\begin{collect*}{appendix-suspect}{
    \begin{lemma}\mylabel{Lem}{prop:correctness-deviator-game}
      Let $\calG$ be a game and $\sigma_\Pl$ be a strategy profile and $\sigma_\shortEve = \projstrat(\sigma_\Pl)$ the associated strategy in the deviator game.
      \begin{enumerate}
      \item If $\rho \in \Out_{\devgame} (\sigma_\shortEve)$, then $\dev(\projpath(\rho),\sigma_\Pl) = \devlimit(\rho)$.  
      \item If $\rho \in \Out_\calG$ and $\rho'= ((\rho_i,\dev(\rho_{\le i},\sigma_\Pl)) \cdot (\sigma_\Pl(\rho_{\le i}), \act_i(\rho)))_{i\in \N}$ then $\rho' \in \Out_\devgame(\sigma_\shortEve)$
      \end{enumerate}
  \end{lemma}
    }{}{}{
\begin{proof}[Proof of 1]
  We prove that for all $i$, $\dev(\projpath(\rho)_{\le i}, \sigma_\Pl) = \projdeux(\rho_{\le i})$, which implies the property.
  The property holds for $i=0$, since initially both sets are empty.
  Assume now that it holds for $i\ge 0$.
  \begin{align*}
    &\dev(\projpath(\rho)_{\le i+1}, \sigma_\Pl)\\
    &= \dev(\projpath(\rho)_{\le i},\sigma_\Pl) \cup \dev(\sigma_\Pl(\projpath(\rho)_{\le i}), \projact(\act_{i+1}(\rho))) & \text{(by definition of deviators)}\\
    &= \projdeux(\rho_{\le i}) \cup \dev(\sigma_\Agt(\projact(\rho)_{\le i}), \projact(\act_{i+1}(\rho))) & \text{(by induction hypothesis)}\\
    &= \projdeux(\rho_{\le i}) \cup  \dev(\sigma_\shortEve(\rho_{\le i}), \projact(\act_{i+1}(\rho))) & \text{(by definition of $\sigma_\shortEve$)}\\
    &= \projdeux(\rho_{\le i}) \cup  \dev(\act_{i+1}(\rho)) & \text{(by assumption $\rho \in \Out_{\devgame}(\sigma_\shortEve)$)}\\
    &= \projdeux(\rho_{\le i+1}) & \text{(by construction of $\devgame$)}
  \end{align*}
  Which concludes the induction. 
\end{proof}

\begin{proof}[Proof of 2]
The property is shown by induction. 
It holds for the initial state.
  Assume it is true until index $i$, then
  \begin{align*} 
    \Tab'&(\rho'_i,\sigma_\shortEve(\rho'_{\le i}),\act_{i}(\rho)) \\
    &= \Tab'((\rho_i,\dev(\rho_{\le i},\sigma_\Agt)),\sigma_\shortEve(\rho'_{\le i}),\act_{i}(\rho))  & \text{(by definition of $\rho'$)}\\
    &=(\Tab(\rho_i,\act_{i}(\rho)), \dev(\rho_{\le i},\sigma_\Agt)\cup \dev(\sigma_\shortEve(\rho'_{\le i}), \rho_{i+1})) & \text{(by construction of $\Tab'$)}\\
    &=(\rho_{i+1}, \dev(\rho_{\le i},\sigma_\Agt)\cup \dev(\sigma_\shortEve(\rho'_{\le i}), \rho_{i+1}))  & \text{(since $\rho$ is an outcome of the game)}\\
    &=(\rho_{i+1}, \dev(\rho_{\le i},\sigma_\Agt)\cup \dev(\sigma_\Agt(\rho_{\le i}), \rho_{i+1})) & \text{(by construction of $\sigma_\shortEve$)}\\
    &=(\rho_{i+1}, \dev(\rho_{\le i+1},\sigma_\Agt)) & \text{(by definition of deviators)}\\
    & = \rho'_{i+1}
  \end{align*}
  This shows that $\rho'$ is an outcome of $\sigma_\shortEve$.
\end{proof}
}\end{collect*}


\subsection{Objectives of the deviator game}

We now show how to transform equilibria notions into objectives of the deviator game.
These objectives are defined so that winning strategies correspond to equilibria of the original game.
First, we define an objective~$\Omega(C,A,G)$ in the following lemma, such that a profile which ensures some quantitative goal $G\subseteq \mathbb{R}$ in $\calG$ against coalition $C$ corresponds to a winning strategy in the deviator game.

\begin{collect*}{appendix-suspect}{
\begin{lemma}\label{lem:objective-omega}
  Let $C\subseteq \Agt$ be a coalition, $\sigma_\Agt$ be a strategy profile, $\Goal \subseteq \mathbb{R}$ and $A$ a player.
  We have that for all strategies $\sigma'_C$ for coalition $C$,
  $\payoff_A(\sigma_{-C},\sigma'_C) \in \Goal$ if, and only if, $\projstrat(\sigma_\Agt)$ is winning in $\devgame$ for objective $\Omega(C,A,\Goal) = \{ \rho \mid \delta(\rho) \subseteq C \Rightarrow \payoff_A(\projpath(\rho)) \in \Goal \}$.
\end{lemma}
}{}{}{
\begin{proof}
  \fbox{$\Rightarrow$}
  Let $\rho$ be an outcome of $\sigma_\shortEve=\projstrat(\sigma_\Agt)$.
  By Lem.~\ref{prop:correctness-deviator-game}, we have that $\devlimit(\rho) = \dev(\projpath(\rho),\sigma_\Agt)$.
  By Lem.~\ref{lem:deviator-path}, $\projpath(\rho)$ is the outcome of $(\sigma_{-\devlimit(\rho)},\sigma'_{\devlimit(\rho)})$ for some $\sigma'_{\devlimit(\rho)}$.
  If $\devlimit(\rho) \subseteq C$, then $\payoff_A(\projpath(\rho)) = \payoff_A(\sigma_{-C},\sigma_{C\setminus\devlimit(\rho)}, \sigma'_{\devlimit(\rho)}) = \payoff_A(\sigma_{-C},\sigma''_{C})$ where $\sigma''_A = \sigma'_A$ if $A \in \devlimit(\rho)$ and $\sigma_A$ otherwise.
  By hypothesis, this payoff belongs to $\Goal$.
  This holds for all outcomes~$\rho$ of $\sigma_\shortEve$, thus $\sigma_\shortEve$ is a winning strategy for $\Omega(C,A,\Goal)$.

  \medskip
  \fbox{$\Leftarrow$}
  Assume $\sigma_\shortEve = \projstrat(\sigma_\Agt)$ is a winning strategy in \devgame for $\Omega(C,A,\Goal)$.
  Let $\sigma'_C$ be a strategy for $C$ and $\rho$ the outcome of $(\sigma'_{C},\sigma_{-{C}})$.
  By Lem.~\ref{lem:deviator-path}, $\dev(\rho,\sigma_\Agt) \subseteq C$.
  By Lem.~\ref{prop:correctness-deviator-game}, $\rho'= (\rho_j,\dev(\rho_{\le j},\sigma_\Agt))_{j\in \N}$ is an outcome of $\sigma_\shortEve$.
  We have that $\devlimit(\rho') = \dev(\rho,\sigma_\Agt) \subseteq C$.
  Since $\sigma_\shortEve$ is winning, $\rho$ is such that $\payoff_A(\projpath(\rho)) \in \Goal$.
  Since $\payoff_{A}(\projun(\rho')) = \payoff_{A}(\rho)$, this shows that for all strategies $\sigma'_C$, $\payoff_A(\sigma_{-C},\sigma'_C) \in \Goal$ 
\end{proof}
}\end{collect*}

This lemma makes it easy to characterise the different kinds of equilibria, using objectives in $\devgame$.
For instance, we define a \newdef{resilience objective} where if there are more than $k$ deviators then \eve has nothing to do; if there are exactly $k$ deviators then she has to show that none of them gain anything; and if there are less than $k$ then no player at all should gain anything.
This is because if a new player joins the coalition, its size remains smaller or equal to $k$.
Similar characterisations for immune and robust equilibria lead to the following theorem.

\begin{theorem}\label{thm:dev-correct}
  Let \calG be a concurrent game, $\sigma_\Agt$ a strategy profile in \calG, $p = \payoff(\Out(\sigma_\Agt))$ the payoff profile of $\sigma_\Agt$, $k$ and $t$ integers, and $r$ a rational.
  \begin{itemize}
  \item
    The strategy profile $\sigma_\Agt$ is $k$-resilient if, and only if, strategy~$\projstrat(\sigma_\Agt)$ is winning in \devgame for the \emph{resilience objective} $\calRe(k,p)$ where
 $\calRe(k,p)$ is defined by: 
    $\calRe(k,p) = $ $ \{ \rho \mid ~ |\devlimit(\rho)| > k \} $  $\cup \{ \rho \mid  ~ |\devlimit(\rho)| = k \land \forall A \in \devlimit(\rho).\ \payoff_{A}(\projpath(\rho)) \le p(A)\}$ $\cup \{ \rho \mid  ~ |\devlimit(\rho)| < k \land \forall A \in \Agt.\ \payoff_{A}(\projpath(\rho)) \le p(A) \} $
  \item
    The strategy profile $\sigma_\Agt$ is $(t,r)$-immune if, and only if, strategy~$\projstrat(\sigma_\Agt)$ is winning for the \newdef{immunity objective} $\calI(t,r,p)$ 
    $\calI(t,r,p)$ is defined by:
    $\calI(t,r,p) =$ $\{ \rho \mid |\devlimit(\rho)| > t \} 
      \cup   \{ \rho \mid  ~ \forall A \in \Agt \setminus \devlimit(\rho).\  p(A) - r \le \payoff_{A}(\projun(\rho)) \} $
  \item
    The strategy profile~$\sigma_\Agt$ is a $(k,t,r)$-robust profile in $\calG$ if, and only if, $\projstrat(\sigma_\Agt)$ is winning for the \emph{robustness objective} $\calR(k,t,r,p)= \calRe(k,p) \cap \calI(t,r,p)$.
  \end{itemize}
\end{theorem}


\begin{collect}{appendix-suspect}{}{}
\subsection{Proof of Thm.~\ref{thm:dev-correct}}\label{sec:proof}
The proof of the theorem relies on the two following lemmas. 
The first one shows the correctness of the resilience objective. 
The second lemma shows the correctness of the immunity objective, its proof follows the same ideas than the first and can be found in the appendix.
\begin{lemma} \label{lem:obj-resilience}{}
  Let \calG be a concurrent game and $\sigma_\Agt$ a strategy profile in \calG.
  The strategy profile $\sigma_\Agt$ is $k$-resilient if, and only if, strategy~$\projstrat(\sigma_\Agt)$ is winning in \devgame for objective $\calRe(k,p)$ where $p = \payoff(\sigma_\Agt)$.
\end{lemma}
\begin{proof}
  By Lem.~\ref{lem:objective-omega}, $\sigma_\Agt$ is $k$-resilient if, and only if, for each coalition $C$ of size smaller than $k$, and each player $A$ in $C$, $\projstrat(\sigma_\Agt)$ is winning for $\Omega(C,A,\rbrack-\infty,\payoff_A(\sigma_\Agt)\rbrack)$.
  We will thus in fact show that for each coalition $C$ of size smaller than $k$, and each player $A$ in $C$, $\projstrat(\sigma_\Agt)$ is winning for $\Omega(C,A,\rbrack-\infty,\payoff_A(\sigma_\Agt)\rbrack)$ if, and only if, $\projstrat(\sigma_\Agt)$ is winning for $\calRe(k,p)$.

  \fbox{$\Rightarrow$}
  Let $\rho$ be an outcome of $\projstrat(\sigma_\Agt)$.
  \begin{itemize}
  \item If $|\devlimit(\rho)| > k$, then $\rho$ is in $\calRe(k,p)$ by definition.
  \item If $|\devlimit(\rho)| = k$, then for all $A\in \devlimit(\rho)$, $\payoff_A(\projpath(\rho)) \in \rbrack-\infty,p(A)\rbrack$ because $\projstrat(\sigma_\Agt)$ is winning for $\Omega(\devlimit(\rho),A,\rbrack-\infty,p(A)\rbrack)$.
    Therefore $\rho$ is in $\calRe(k,p)$.
  \item If $|\devlimit(\rho)| < k$, then for all $A\in \Agt$, $C= \devlimit(\rho)\cup \{A\}$ is a coalition of size smaller than $k$, and $\payoff_A(\projpath(\rho)) \in \rbrack-\infty,p(A)\rbrack$ because $\projstrat(\sigma_\Agt)$ is winning for $\Omega(C,A,\rbrack-\infty,p(A)\rbrack)$.
    Therefore $\rho$ is in $\calRe(k,p)$.
  \end{itemize}
  This holds for all outcomes $\rho$ of $\projstrat(\sigma_\Agt)$ and shows that $\projstrat(\sigma_\Agt)$ is winning for $\calRe(k,p)$.

  \medskip
  \fbox{$\Leftarrow$}
  We now show that $\projstrat(\sigma_\Agt)$ is winning for $\Omega(C,A,\rbrack-\infty,p(A)\rbrack)$ for each coalition $C$ of size smaller or equal to $k$ and player $A$ in $C$.
  Let $\rho$ be an outcome of $\projstrat(\sigma_\Agt)$.
  Let $p$ be such that strategy $\projstrat(\sigma_\Agt)$ is winning for $\calRe(k,p)$.
  We have $\rho \in \calRe(k,p)$.
  We show that $\rho$ belongs to $\Omega(C,A,\rbrack -\infty, p(A)\rbrack)$:
  \begin{itemize}
  \item If $\devlimit(\rho) \not\subseteq C$, then $\rho \in \Omega(C,A,\rbrack -\infty, p(A)\rbrack)$ by definition.
  \item If $\devlimit(\rho) \subseteq C$ and $|\devlimit(\rho)| = k$, then $\dev(\rho) = C$.
    Since $\rho \in \calRe(k,p)$, for all $A\in C$, $\payoff_A(\rho) \le p(A)$ and therefore $\payoff_A(\rho) \in \rbrack -\infty, p(A) \rbrack$.
    Hence $\rho \in \Omega(C,A,\rbrack -\infty, p(A)\rbrack)$.
  \item If $\devlimit(\rho) \subseteq C$ and $|\devlimit(\rho)| < k$, then since $\rho \in \calRe(k,p)$, for all $A\in \Agt$, $\payoff_A(\rho) \le p(A)$.
    Therefore $\rho \in \Omega(C,A,\rbrack -\infty, p(A)\rbrack)$.
  \end{itemize}
  This holds for all outcomes $\rho$ of $\projstrat(\sigma_\Agt)$ and shows it is winning for $\Omega(C,A,\rbrack-\infty,p(A)\rbrack)$ for each coalition $C$ and player $A$ in $C$, which shows that $\sigma_\Agt$ is $k$-resilient.
\end{proof}


\begin{lemma}\label{lem:obj-immunity}
  Let \calG be a concurrent game and $\sigma_\Agt$ a strategy profile in \calG.
  The strategy profile $\sigma_\Agt$ is $(t,r)$-immune if, and only if, strategy~$\projstrat(\sigma_\Agt)$ is winning for objective $\calI(t,r,p)$ where $p = \payoff(\sigma_\Agt)$.
\end{lemma}
\begin{proof}
  By Lem.~\ref{lem:objective-omega}, $\sigma_\Agt$ is $(t,r)$-immune if, and only if, for each coalition $C$ of size smaller than $t$, and each player $A$ not in $C$, $\projstrat(\sigma_\Agt)$ is winning for $\Omega(C,A,\lbrack \payoff_A(\sigma_\Agt) - r,+\infty\lbrack)$.
  We will thus in fact show that for each coalition $C$ of size smaller than $t$, and each player $A$ not in $C$, $\projstrat(\sigma_\Agt)$ is winning for $\Omega(C,A,\lbrack \payoff_A(\sigma_\Agt) -r ,+\infty\lbrack)$ if, and only if, $\projstrat(\sigma_\Agt)$ is winning for $\calRe(t,r,p)$.

  \fbox{$\Rightarrow$}
  Let $\rho$ be an outcome of $\projstrat(\sigma_\Agt)$.
  \begin{itemize}
  \item If $|\devlimit(\rho)| > t$, then $\rho$ is in $\calI(t,r,p)$ by definition.
  \item If $|\devlimit(\rho)| \le t$, then $C = \devlimit(\rho)$ is a coalition of size smaller than $t$.
    As a consequence, for all $A \not\in \devlimit(\rho)$, $\rho$ is winning for $\Omega(C,A,\lbrack p_A -r , +\infty \lbrack)$.
    By definition of $\Omega$, we have $\payoff_A(\rho) \ge p_A - r$.
    Thus $\rho$ is in $\calI(t,r,p)$.
  \end{itemize}

  \fbox{$\Leftarrow$}
  We now show that $\projstrat(\sigma_\Agt)$ is winning for $\Omega(C,A,\lbrack \payoff_A(\sigma_\Agt) -r ,+\infty\lbrack)$ for each coalition $C$ of size smaller than $t$ and player $A$ not in $C$.  
  Let $\rho$ be an outcome of $\projstrat(\sigma_\Agt)$.
  Let $p$ be such that strategy $\projstrat(\sigma_\Agt)$ is winning for $\calI(t,r,p)$.
  We have $\rho \in \calI(t,r,p)$.
  We show that $\rho$ belongs to $\Omega(C,A,\lbrack p(A)-r, +\infty \lbrack)$:
  \begin{itemize}
  \item If $\devlimit(\rho) \not\subseteq C$, then $\rho \in \Omega(C,A,\lbrack p(A)-r, +\infty \lbrack)$ by definition.
  \item If $\devlimit(\rho) \subseteq C$, then since $\rho \in \calI(t,r,p)$, for all $A \not\in C$, $p(A) -r \le \payoff_A(\projpath(\rho))$.
    Therefore $\payoff_A(\rho) \in \lbrack p(A)-r, +\infty \lbrack$ and $\rho \in \Omega(C,A,\lbrack p(A)-r, +\infty \lbrack)$.
  \end{itemize}
  This holds for all outcomes $\rho$ of $\projstrat(\sigma_\Agt)$ and shows it is winning for $\Omega(C,A,\lbrack p(A)-r, +\infty\lbrack)$ for each coalition $C$ and player $A$ in $C$, which shows that $\sigma_\Agt$ is $(t,r)$-immune.
\end{proof}


\begin{lemma}
  Let \calG be a concurrent game and $\sigma_\Agt$ a strategy profile in \calG.
  The strategy profile~$\sigma_\Agt$ is a $(k,t,r)$-robust profile in $\calG$ if, and only if, the associated strategy of \eve is winning for the objective $\calR(k,t,r,\Out(\sigma_\Agt)) = \calRe(k,p) \cap \calI(t,r,p)$ where $p = \payoff(\sigma_\Agt)$.
  \end{lemma}
\begin{proof}
  This is a simple consequence of Lem.~\ref{lem:obj-resilience} and Lem.~\ref{lem:obj-immunity}. 
  Let $\sigma_\Agt$ be a $(k,t,r)$-robust strategy profile.
  It is $k$-resilient, so $\projstrat(\sigma_\Agt)$ is winning the resilience objective.
  It is also $(t,r)$-immune, so $\projstrat(\sigma_\Agt)$ is winning the immunity objective.
  Therefore any outcome of $\sigma_\Agt$ is in the intersection, and $\sigma_\Agt$ ensures the robustness objective.
  
  In the other direction, assume $\projstrat(\sigma_\Agt)$ wins the robustness objective.
  Then $\sigma_\Agt$ wins both the $k$-resilience objective and the $(t,r)$-immunity objective.
  Using lemmas~\ref{lem:obj-resilience} and \ref{lem:obj-immunity}, $\sigma_\Agt$ is $k$-resilient and $(t,r)$-immune; it is therefore $(k,t,r)$-robust.
\end{proof}

\end{collect}


\section{Reduction to multidimensional mean-payoff objectives}\label{sec:mean-payoff}

We first show that the deviator game reduces the robustness problem to a winning strategy problem in multidimensional mean-payoff games.
We then solve this by requests to the polyhedron value problem of~\cite{BR15}.

\subsection{Multidimensional objectives}\label{sec:def-multidim}


\begin{mydefinition}[Multidimensional mean-payoff objective]
  Let $\mathcal{G}$ be a two-player game, $v \colon \Stat \mapsto \mathbb{Z}^{d}$ a multidimensional weight functions and $I,J \subseteq \lsem 1, d \rsem$\footnote{We write $\lsem i, j\rsem$ for the set of integers $\{ k \in \mathbb{Z} \mid i \le k \le j\}$.} a partition of $\lsem 1, d\rsem$ (i.e. $I \uplus J = \lsem 1, d \rsem$).
  We say that \eve \emph{ensures threshold} $u\in \mathbb{R}^{d}$ if she has a strategy $\sigma_\exists$ such that all outcomes $\rho$ of $\sigma_\exists$ are such that 
  for all $i \in I$, $\MP_{v_i}(\rho) \ge u_i$ and for all $j \in J$, 
  $\overline\MP_{v_{j}}(\rho) \ge u_{j}$, 
  where \(\MPSup_{v_j}(\rho) = \limsup_{n \rightarrow \infty} \frac{1}{n} \sum_{0 \le k\le n} v_j(\rho_{k}). \)
  That is, for all dimensions~$i\in I$, the limit inferior of the average of $v_i$ is greater than $u_i$ and for all dimensions $j\in J$ the limit superior of $v_j$ is greater than $u_j$.
\end{mydefinition}

We consider two decision problems on these games:
\begin{inparaenum}
\item The \emph{value problem}, asks given $\langle \mathcal{G},v,I,J\rangle$ a game with multidimensional mean-payoff objectives, and $u\in \mathbb{R}^d$, whether \eve can ensure $u$.
\item The \emph{polyhedron value problem}, asks given $\langle \mathcal{G},v,I,J\rangle$ a game with multidimensional mean-payoff objectives, and $(\lambda_1,\dots,\lambda_n)$ a tuple of linear inequations, whether there exists a threshold $u$ which \eve can ensure and that satisfies the inequation $\lambda_i$ for all $i$ in $\lsem 1 ,d \rsem$.
  We assume that all linear inequations are given by a tuple $(a_1,\dots,a_d,b)\in\mathbb{Q}^{d+1}$ and that a point $u\in \mathbb{R}^d$ satisfies it when $\sum_{i\in \lsem 1 , d \rsem} a_i \cdot u_i \ge b$.
\end{inparaenum}
The value problem was showed to be \co\NP-complete~\cite{velner12} while the polyhedron value problem is $\Sigma_2$\P-complete~\cite{BR15}.
Our goal is now to reduce our robustness problem to a polyhedron value problem for some well chosen weights.

\medskip

In our case, the number $d$ of dimensions will be equal to $4 \cdot |\Agt|$. 
We then number players so that $\Agt = \{ A_1, \dots , A_{|\Agt|}\}$.
Let $W = \max\{ |w_{i}(s)| \mid A_i \in \Agt, s \in \Stat\}$ be the maximum constant occurring in the weights of the game, notice that for all players $A_i$ and play $\rho$, $ - W - 1 < \MP_{i}(\rho) \le W$.
We fix parameters $k$, $t$ and define our weight function $v\colon \Stat \mapsto \mathbb{Z}^{d}$.
Let $i\in \lsem 1 ,|\Agt| \rsem$, the weights are given for $(s,D) \in \Stat \times 2^\Agt$ by: 
\begin{enumerate}
\item if $|D|\le t$ and $A_i\not\in D$, then $v_i(s,D) = w_{A_i}(s)$;
\item if $|D|> t$ or $A_i\in D$, then $v_i(s,D) = W$;
\item \label{it:small-k} if $|D|<k$, then for all $A_i \in \Agt$, $v_{|\Agt|+i}(s,D) = - w_{A_i}(s)$;
\item \label{it:equal-k} if $|D|=k$ and $A_i\in D$, then $v_{|\Agt|+i}(s,D) = - w_{A_i}(s)$;
\item \label{it:greater-k} if $|D| > k$ or $A_i \not\in D$ and $|D| = k$, then $v_{|\Agt|+i}(s,D) = W$.
\item \label{it:Deqvarnothing} if $D = \varnothing$ then $v_{2\cdot|\Agt|+i}(s,D)  = w_{A_i}(s)= -v_{3\cdot|\Agt|+i}(s,D)$;
\item \label{it:Dnevarnothing} if $D \ne \varnothing$ then $v_{2\cdot|\Agt|+i}(s,D)  = W = v_{3\cdot|\Agt|+i}(s,D)$;
\end{enumerate}
We take $I = \lsem 1 , |\Agt|\rsem \cup \lsem 2\cdot |\Agt|+1, 3 \cdot |\Agt| \rsem$ and $J = \lsem |\Agt|+1, 2\cdot |\Agt|\rsem \cup \lsem 3 \cdot |\Agt|+1, 4\cdot |\Agt|\rsem$.
Intuitively, the components $\lsem 1,|\Agt|\rsem$ are used for immunity, the components $\lsem |\Agt|+1, 2\cdot|\Agt|\rsem$ are used for resilience
and components $\lsem 2\cdot |\Agt|+1, 4\cdot |\Agt|\rsem$ are used to constrain the payoff in case of no deviation.

\subsection{Correctness of the objectives for robustness}

Let $\calG$ be a concurrent game, $\rho$ a play of $\devgame$ and $p\in \mathbb{R}^\Agt$ a payoff vector.
The following lemma links the weights we chose and our solution concepts.

\begin{lemma}\label{lem:correctness-objectives-robustness}
  Let $\rho$ be a play, $\sigma_\Agt$ a strategy profile and $ p =\payoff(\sigma_\Agt)$.
  \begin{itemize}
  \item $\rho$ satisfies objective $\devlimit(\rho) = \varnothing \Rightarrow \MP_{A_i}(\rho) = p_i$ if, and only if, $\MP_{2\cdot v_{|\Agt|+i}}(\rho) \ge p(A_i)$ and $\MPSup_{3\cdot v_{|\Agt|+i}}(\rho) \ge -p(A_i)$.
  \item
    If $\rho$ is an outcome of $\projstrat(\sigma_\Agt)$ then $\rho$ satisfies objective $\calRe(k,p)$ if, and only if, for all agents $A_i$, $\MPSup_{v_{|\Agt|+i}}(\rho) \ge -p(A_i)$.
  \item 
    If $\rho$ is an outcome of $\projstrat(\sigma_\Agt)$, then $\rho$ satisfies objective $\calI(t,r,p)$ if, and only if, for all agents $A_i$, $\MP_{v_i}(\rho) \ge p(A_i)-r$.
  \item
     If~$\rho$ is an outcome of $\projstrat(\sigma_\Agt)$ with $\payoff(\sigma_\Agt) = p$, then play $\rho$ satisfies 
  objective $\calR(k,t,r,p)$ if, and only if, for all agents $A_i$, $\MP_{v_i}(\rho) \ge p(A_i)-r$ and $\MPSup_{v_{|\Agt| + i}}(\rho) \ge -p(A_i)$.
  \end{itemize}
\end{lemma}

\begin{collect}{appendix-mean-payoff}{}{}
  \subsection{Proof of Lem.~\ref{lem:correctness-objectives-robustness}}
\end{collect}

\begin{collect}{appendix-mean-payoff}{}{}
\begin{lemma}\label{lem:mp-payoff}
  Let $\rho$ be a play.
  It satisfies objective $\devlimit(\rho) = \varnothing \Rightarrow \MP_{A_i}(\rho) = p_i$ if, and only if, $\MP_{2\cdot v_{|\Agt|+i}}(\rho) \ge p(A_i)$ and $\MPSup_{3\cdot v_{|\Agt|+i}}(\rho) \ge -p(A_i)$.
\end{lemma}
\begin{proof}
  We distinguish two cases according to whether $\devlimit(\rho)$ is empty.
  \begin{itemize}
  \item If $\devlimit(\rho) \ne \varnothing$, the implication holds and we have that after some point in the execution $\projdeux(\rho) \ne \varnothing$.
    By item~\ref{it:Dnevarnothing} of the definition of $v$, the average weight on dimensions $2\cdot|\Agt|+i$ and $3\cdot |\Agt|+i$ will tend to $W$, which is greater than $p(A_i)$ and $-p(A_i)$.
    Therefore the equivalence holds.
  \item If $\devlimit(\rho) = \varnothing$, then along all the run the $D$ component is empty.
    By item~\ref{it:Deqvarnothing} of the definition of $v$,
    $\MP_{v_{2\cdot |\Agt|+i}}(\rho) = \MP_{w_{A_i}} (\rho)$ and $\MPSup_{v_{3\cdot |\Agt|+i}}(\rho) = - \MP_{w_{A_i}} (\rho)$.
    Therefore $\MP_{A_i}(\rho) = p_i$ is equivalent to the fact $\MP_{2\cdot v_{|\Agt|+i}}(\rho) \ge p(A_i)$ and $\MPSup_{3\cdot v_{|\Agt|+i}}(\rho) \ge -p(A_i)$.
  \end{itemize}
\end{proof}
\end{collect}

\begin{collect}{appendix-mean-payoff}{}{}
\begin{lemma}\label{lem:mp-resilient}
  If $\rho$ is an outcome of $\projstrat(\sigma_\Agt)$ with $\payoff(\sigma_\Agt) = p$, then play~$\rho$ satisfies objective $\calRe(k,p)$ if, and only if, for all agents $A_i$, $\MPSup_{v_{|\Agt|+i}}(\rho) \ge -p(A_i)$.
\end{lemma}
\begin{proof}
  First notice the following equivalence:
  \begin{align*}
    \payoff_{A_i}(\rho) \le p(A_i) & \Leftrightarrow  \lim\inf \frac{w_i(\rho_{\le n}) }{n} \le p(A_i)
    \Leftrightarrow \lim\sup -\frac{w_i(\rho_{\le n}) }{n} \ge -p(A_i) \\
  \end{align*}

  \fbox{$\Rightarrow$}
  Let $A_i$ be a player, and assume $\rho \in \calRe(k,p)$. 
  We distinguish three cases based on the size of $\devlimit(\rho)$:
  \begin{itemize}
  \item If $|\devlimit(\rho)| < k$ then for all indices $j$, $|\dev(\rho_{\le j'})|<k$ .
    Therefore $v_{|\Agt|+i}(\rho_{j}) = - w_{A_i}(\rho_{j})$ (item \ref{it:small-k} of the definition).
    Then as $\rho$ is in $\calRe(k,p)$, $\payoff_{A_i}(\rho) \le p(A_i)$ and therefore $\overline\MP_{v_{|\Agt|+i}}(\rho) \ge -p(A_i)$
  \item If $|\devlimit(\rho)|=k$, then we distinguish two cases: 
    \begin{itemize}
    \item If $A_i\not\in \devlimit(\rho)$, then there is a $j$ such that for all $j'\ge j$, $|\dev(\rho_{\le j})| = k$ and $A_i \not\in \dev(\rho_{j})$.
      Therefore for all $j' \ge j$ we have that $v_{|\Agt|+i}(\rho_{j'}) = W$ (item \ref{it:greater-k} of the definition).
      Since $p(A_i) \ge -W$, $\overline\MP_{v_{|\Agt|+i}}(\rho) \ge -p(A_i)$.
    \item
      Otherwise $A_i\in \devlimit(\rho)$, then there is a $j$ such that for all $j'\ge j$, $A_i \in \dev(\rho_{\le j'})$.
      Therefore for all $j' \ge j$ we have that $v_{|\Agt|+i}(\rho_{j'}) = - w_{A_i}(\rho_{j'})$ (item \ref{it:equal-k} of the definition).
      Then $\overline\MP_{v_{|\Agt|+i}}(\rho) = \lim\sup - \frac{w_i(\rho_{\le n})}{n}$.
      Then as $\rho$ satisfies $\calRe(k,p)$
      $\payoff_{A_i}(\rho) \le p(A_i)$ and therefore using the equivalence at the beginning of this proof $\overline\MP_{v_{|\Agt|+i}}(\rho) \ge -p(A_i)$.
    \end{itemize}
  \item Otherwise $|\devlimit(\rho)| > k$.
    Then, there is some index~$j$ such that either $|\dev(\rho_{\le j})| > k$ or  $|\dev(\rho_{\le j})| = k \land A\not\in \dev(\rho_{\le j})$.
    Then, by monotonicity of $\dev$ along $\rho$, for all $j'\ge j$, $v_{|\Agt|+i}(\rho_{j'}) = W$ (item \ref{it:greater-k} of the definition).
    Since $p(A_i) \ge -W$, $\overline\MP_{v_{|\Agt|+i}}(\rho) \ge -p(A_i)$.
  \end{itemize}

  \fbox{$\Leftarrow$}
  Now assume that for all players $A_i$, $\MPSup_{v_{|\Agt|+i}}(\rho) \ge - p(A_i)$.
  \begin{itemize}
  \item If $|\devlimit(\rho)|<k$, therefore for all $i$ and $j$ we have that $v_{|\Agt|+i}(\rho_j) = - w_{A_i}(\rho_j)$ then $\overline\MP_{v_{|\Agt|+i}}(\rho) = \lim\sup - \frac{w_i(\rho_{\le n})}{n}$.
    Thus using the equivalence at the beginning of this proof $\payoff_{A_i}(\rho) \le p(A_i)$ for all $A_i$.
  \item If $|\devlimit(\rho)|=k$.
    Let $A_i$ be a player in $\devlimit(\rho)$.
    Then for all $j$, either $|\dev(\rho_{\le j})|<k$ or $A_i\in \dev(\rho_{\le j})$.
    Therefore for all $j$ we have that $v_{|\Agt|+i}(\rho_j) = - w_{A_i}(\rho_j)$ then $\overline\MP_{v_{|\Agt|+i}}(\rho) = \lim\sup - \frac{w_i(\rho_{\le n})}{n}$.
    Thus using the equivalence at the beginning of this proof $\payoff_{A_i}(\rho) \le p(A_i)$.
    This being true for all players in $\devlimit(\rho)$ shows that $\rho$ belongs to $\calRe(k,p)$.
  \item
    Otherwise $|\devlimit(\rho)| > k$ and then $\rho \in \calRe(k,p)$ by definition of $\calRe(k,p)$.
  \end{itemize}
\end{proof}
\end{collect}

\begin{collect}{appendix-mean-payoff}{}{}
We now show the immunity part.
\begin{lemma}\label{lem:mp-immune}
  If $\rho$ is an outcome of $\projstrat(\sigma_\Agt)$ with $\payoff(\sigma_\Agt) = p$, then play~$\rho$ satisfies objective $\calI(t,r,p)$ if, and only if, for all agents $A_i$, $\MP_{v_i}(\rho) \ge p(A_i)-r$.
\end{lemma}
\begin{proof}
  \fbox{$\Rightarrow$}
  Let $A_i$ be a player and assume $\rho\in \calI(t,r,p)$.
  We distinguish two cases:
  \begin{itemize}
  \item If $|\devlimit(\rho)|\le t \land A_i \not\in \devlimit(\rho)$, then for all indices $j$, $v_i(\rho_{j}) = w_{A_i}(\rho_j)$. 
    Therefore $\MP_{v_i}(\rho) = \payoff_{A_i}(\projun(\rho))$.
    Then as $\rho$ satisfies $\calI(t,r,p)$, $p(A_i) -r \le \payoff_{A_i}(\projun(\rho)) = \MP_{v_i}(\rho)$.
  \item
    Otherwise there is some index~$j$ such that either $|\dev(\rho_{\le j})| > t$ or $A \in \dev(\rho_{\le j})$.
    Then, by monotonicity of $\dev$ along $\rho$, for all $j'\ge j$, $v_i(\rho_{j'}) = W \ge p(A_i)$.
    Hence $\MP_{v_i}(\rho) \ge p(A_i)$.
  \end{itemize}

  \fbox{$\Leftarrow$}
  Assume that for all players~$A_i$, $\MP_{v_i}(\rho) \ge p(A_i) - r$.
  \begin{itemize}
  \item
    If $|\devlimit(\rho)|\le t$.
    Let $A_i\not\in \devlimit(\rho)$, then for all $j$, $|\dev(\rho_{\le j})|\le t$ and $A\not\in \dev(\rho_{\le j})$.
    Therefore for all $j$ we have that $v_i(\rho_j) = w_{A_i}(\rho_j)$ and thus
    $\MP_{v_i}(\rho) \ge p(A_i) - r \Leftrightarrow \payoff_{A_i}(\rho) -r \le \payoff_{A_i}(\rho)$.
    This shows that $\rho$ belongs to $\calI(t,r,p)$. 
  \item 
    Otherwise $\devlimit(\rho)| > t$ and $\rho$ belongs to $\calI(t,r,p)$ by definition of $\calI(t,r,p)$.
  \end{itemize}
\end{proof}
\end{collect}

\begin{collect}{appendix-mean-payoff}{}{}
We now join the two preceding result to talk about the robustness objective.
\begin{lemma}\label{lem:multidim}
  If~$\rho$ is an outcome of $\projstrat(\sigma_\Agt)$ with $\payoff(\sigma_\Agt) = p$, then play $\rho$ satisfies 
  objective $\calR(k,t,r,p)$ if, and only if, for all agents $A_i$, $\MP_{v_i}(\rho) \ge p(A_i)-r$ and $\MPSup_{v_{|\Agt| + i}}(\rho) \ge -p(A_i)$.
\end{lemma}
\begin{proof}
  \begin{align*}
    \rho \in \calR(k,t,r,p) \Leftrightarrow & \rho \text{ satisfies }\calRe(k,p)\text{ and }\calI(t,r,p) \text{ ~~~(By definition of $\calR(k,t,r,p)$)} \\
     \Leftrightarrow & \rho \in \calRe(k,p)\text{ and } \forall A_i\in \Agt.\ \MP_{v_i}(\rho) \ge p(A_i)-r \text{ ~~(By  Lem.~\ref{lem:mp-immune}) } \\
     \Leftrightarrow & \forall A_i\in \Agt.\ \MPSup_{v_{|\Agt| + i}}(\rho) \ge -p(A_i) \\ & \text{ and } \forall A_i\in \Agt.\ \MP_{v_i}(\rho) \ge p(A_i)-r \text{ ~~~(By  Lem.~\ref{lem:mp-resilient}) } 
  \end{align*}
\end{proof}
\end{collect}

Putting together this lemma and the correspondence between the deviator game and robust equilibria of Thm.~\ref{thm:dev-correct} we obtain the following proposition.
\begin{collect*}{appendix-mean-payoff}{
\begin{lemma}\mylabel{Lem}{lem:correct-mean-robust}
  Let \calG be a concurrent game with mean-payoff objectives.
  There is a $(k,t,r)$-robust equilibrium in $\calG$ 
  if, and only if, 
  for the multidimensional mean-payoff objective given by $v$, $I = \lsem 1 , |\Agt|\rsem \cup \lsem 2\cdot |\Agt|+1, 3 \cdot |\Agt| \rsem$ and $J = \lsem |\Agt|+1, 2\cdot |\Agt|\rsem \cup \lsem 3 \cdot |\Agt|+1, 4\cdot |\Agt|\rsem$, 
  there is a payoff vector $p$ such that \eve can ensure threshold $u$ in $\devgame$, where for all $i \in \lsem 1,|\Agt|\rsem$, $u_{i} = p(A_i) -r$, $u_{|\Agt|+i} = -p(A_i)$, $u_{2\cdot |\Agt|+i} = p(A_i)$, and $u_{3\cdot |\Agt|+i} = -p(A_i)$.
\end{lemma}
}{}{}{
\begin{proof}
  \fbox{$\Rightarrow$}
  Let $\sigma_\Agt$ be a robust equilibrium, using Thm.~\ref{thm:dev-correct}
  $\projstrat(\sigma_\Agt)$ is a strategy of \eve in $\devgame$ which ensures $\calR(k,t,r,p)$ where $p = \payoff(\Out(\sigma_\Agt))$.
  Let $\rho$ be an outcome of $\projstrat(\sigma_\Agt)$ in $\devgame$.
  We will show that it is above the threshold~$u$ in all dimensions.

  We first show that $\devlimit(\rho) = \varnothing \implies \MP_{A_i}(\rho) = p_i$.
  If $\devlimit(\rho) \ne \varnothing$ this is trivial.
  Otherwise $\devlimit(\rho) = \varnothing$, and by Lem.~\ref{prop:correctness-deviator-game}, there is $\rho'$ such that $\dev(\rho',\sigma_\Agt) = \varnothing$ and $\rho = \projpath(\rho')$.
  Then by Lem.~\ref{lem:deviator-path}, $\rho' = \Out_\calG(\sigma_\Agt)$. 
  Therefore $\rho = \projpath(\Out_\calG(\sigma_\Agt))$, thus $\MP_{A_i}(\rho) = \payoff_{i}(\sigma_\Agt) = p_i$ and the implication holds.
  
  Then by Lem.~\ref{lem:mp-payoff}, we have that $\MP_{v_{2\cdot |\Agt|+i}}(\rho) \ge p(A_i)$ and $\MPSup_{v_{3\cdot |\Agt|+i}}(\rho) \ge -p(A_i)$.
  This shows we ensure the correct thresholds on dimensions in $\lsem 2 \cdot |\Agt|+1, 4\cdot |\Agt|\rsem$.
  Now, by Lem.~\ref{lem:multidim}, for all agents $A_i$, $\MP_{v_i}(\rho) \ge p(A_i)-r$ and $\MPSup_{v_{|\Agt| + i}}(\rho) \ge -p(A_i)$.
  This shows we ensure the correct thresholds on dimensions in $\lsem 1, 2\cdot |\Agt|\rsem$.

  \fbox{$\Leftarrow$}
  In the other direction, let $p$ be a payoff vector such that there exist a strategy $\sigma_\exists$ in $\devgame$ that ensure the threshold $u$.
  We define a strategy profile $\sigma_\Agt$ by induction, given a history $h$ in $\calG$:
  \begin{inparaenum}
  \item if $|h| = 1$, then $h' = (h_0,\varnothing)$;
  \item otherwise we assume $\sigma_\Agt$ has already been defined for histories shorter than $h$, we let $h'= (h_0,\varnothing) \cdot (\sigma_\Pl(h_{\le 0}), \act_0(h)) \cdot \left((h_i,\dev(h_{\le i},\sigma_\Agt)) \cdot (\sigma_\Pl(h_{\le i}), \act_i(h))\right)_{0 < i < |h|-1} \cdot (h,\dev(h,\sigma_\Agt))$.
    Note the similarity with the second point of Lem.~\ref{prop:correctness-deviator-game} and the fact that $\projpath(h') = h$.
  \end{inparaenum}
  We then set $\sigma_\Agt(h) = \sigma_\exists(h')$.

  We show that $\payoff(\Out(\sigma_\Agt)) = p$.
  Consider the strategy~$\sigma_\forall$ of \adam in $\devgame$ that always plays the same move as \eve, the outcome $\rho = \Out_\devgame(\sigma_\exists,\sigma_\forall)$ is such that $\devlimit(\rho)= \varnothing$.
  Since $\sigma_\exists$ ensures the threshold $u$, using Lem.~\ref{lem:mp-payoff}, for all agents $A_i$, $\MP_{A_i}(\rho) = p(A_i)$.
  We now show by induction that $\projpath(\rho)$ is compatible with $\sigma_\Agt$.
  Let $i\in \mathbb{N}$, we assume that the property holds for prefixes of $\projpath(\rho)$ of length less than~$i$.
  We have that $\projact(\rho_i) = \sigma_\Agt(\projpath(\rho)_{\le i}))$ because \adam plays the same move than \eve on this path.
  We also have that $\projun(\rho_{i+1}) = \Tab(\projun(\rho_i),\projact(\rho_i)) =  \Tab(\projun(\rho_i),\sigma_\Agt(\projpath(\rho)_{\le i}))$ by construction of $\devgame$, and therefore $\projpath(\rho_{i+1})$ is compatible with $\sigma_\Agt$.
  Then as $\rho$ is the projection of the outcome of $\sigma_\Agt$, we have that $\payoff(\Out(\sigma_\Agt)) = p$.
  
  Let $\rho$ be an outcome of $\sigma_\exists$.
  By Lem.~\ref{lem:mp-immune}, we have that $\rho$ satisfies the objective $\calI(t,r,p)$, and by Lem.~\ref{lem:mp-resilient} it satisfies $\calRe(k,p)$.
  Using Thm.~\ref{thm:dev-correct}, strategy $\sigma_\Agt$ is an equilibrium in $\calG$ with the payoff $p$.
\end{proof}
}\end{collect*}

\subsection{Formulation of the robustness problem as a polyhedron value problem}
From the previous lemma, we can deduce an algorithm which works by querying the polyhedron value problem.
Given a game $\calG$ and parameters $k,t,r$, we ask whether there exists a payoff $u$ that \eve can ensured in the game $\devgame$ with multidimensional mean-payoff objective given by $v$, $I$, $J$, and such that for all $i \in \lsem 1 , |\Agt|\rsem$, $u_i + r = - u_{|\Agt|+i} = u_{2 \cdot |\Agt| + i} = - u_{3\cdot |\Agt|+i}$.
As we will show in Thm.~\ref{thm:pspace-membership} thanks to Lem.~\ref{lem:correct-mean-robust}, the answer to this question is yes if, and only if, there is a $(k,t,r)$-robust equilibrium.
From the point of view of complexity, however, the deviator game on which we perform the query can be of exponential size compared to the original game.
To describe more precisely the complexity of the problem, we remark by applying the bound of \cite[Thm.~22]{cav15-long}, that given a query, we can find solutions which have a small representation.

\begin{collect*}{appendix-mean-payoff}{
\begin{lemma}\mylabel{Lem}{lem:small-witness}
  If there is a solution to the polyhedron value problem in $\devgame$ then there is one whose encoding is of polynomial size with respect to $\calG$ and the polyhedron given as input.
\end{lemma}
}{}{
\begin{proof}
  If we apply the bound of \cite[Thm.~22]{cav15-long}, to the deviator game, then this shows that if there is a solution there is one of size bounded by $d \cdot P_1(\max\{\size{a_i,b_i} \mid 1 \le i \le k\}, P_6(\size{W}, \size{\Stat \times 2^\Agt}, d), d)$ where $\size{x}$ represents the size of the encoding of the object $x$, $P_1$ and $P_6$ are two polynomial functions, $(a_i,b_i)_{1\le i \le k}$ are the inequations defining the polyhedron, $W$ is the maximal constant occurring in the weights and $d$ is the number of dimension, which in our case is equal to $4 \cdot |\Agt|$.
  The size $\size{\Stat \times 2^\Agt}$ is bounded by $\size{\Stat} + |\Agt|$ (using one bit for each agent in our encoding).
  Thus the global bound is polynomial with respect to the game $\calG$.
\end{proof}
}{}
\end{collect*}
We can therefore enumerate all possible solutions in polynomial space.
To obtain an polynomial space algorithm, we must be able to check one solution in polynomial space as well.
This account to show that queries for the value problem in the deviator game can be done in space polynomial with respect to the original game.
This is the goal of the next section, and it is done by considering small parts of the deviator game called fixed coalition games.


\section{Fixed coalition game}\label{sec:fixed-coalition}
Although the deviator game may be of exponential size, it presents a particular structure.
As the set of deviators only increases during any run, the game can be seen as the product of the original game with a directed acyclic graph (DAG).
The nodes of this DAG correspond to possible sets of deviators, it is of exponential size but polynomial degree and depth.
We exploit this structure to obtain a polynomial space algorithm for the value problem and thus also for the polyhedron value problem and the robustness problem.
The idea is to compute winning states in one component at a time, and to recursively call the procedure for states that are successors of the current component.
We will therefore consider one different game for each component.

We now present the details of the procedure.
For a fixed set of deviator $D$, the possible successors of states of the component $\Stat \times D$ are the states in: \\
$\Succ(D) = \{ \Tab_\calD((s,D), (m_\Agt,m_\Agt')) \mid s \in \Stat, m_\Agt,m'_\Agt\in \Mov(s) \} ~\setminus~ \Stat \times D\}.$ 
Note that the size of $\Succ(D)$ is bounded by $|\Stat| \times |\Tab|$, hence it is polynomial.
Let $u$ be a payoff threshold, we want to know whether \eve can ensure $u$ in $\devgame$, for the multi-dimensional objective defined in Section~\ref{sec:def-multidim}.
A winning path~$\rho$ from a state in $\Stat \times D$ is either: 
\begin{inparaenum}[1)]
\item such that $\devlimit(\rho)= D$;
\item or it reaches a state in $\Succ(D)$ and follow a winning path from there.
\end{inparaenum}
Assume we have computed all the states in $\Succ(D)$ that are winning.
We can stop the game as soon as $\Succ(D)$ is reached, and declare \Eve the winner if the state that is reached is a winning state of $\devgame$.
This process can be seen as a game~$\mathcal{F}(D,u)$, called the \newdef{fixed coalition game}.

In this game the states are those of $(\Stat\times D) \cup \Succ(D)$; transitions are the same than in $\devgame$ on the states of $\Stat \times D$ and the states of $\Succ(D)$ have only self loops.
The winning condition is identical to $\calR(k,t,p)$ for the plays that never leave $(\Stat\times D)$; and for a play that reach some $(s',D') \in\Succ(D)$, it is considered winning \eve has a winning strategy from $(s',D')$ in $\devgame$ and losing otherwise.

In the fixed coalition game, we keep the weights previously defined for states of $\Stat\times D$, and fix it for the states that are not in the same $D$ component by giving a high payoff on states that are winning and a low one on the losing ones.
Formally, we define a multidimensional weight function $v^f$ on $\mathcal{F}(D,u)$ by:
\begin{inparaenum}
\item for all $s \in \Stat$, and all $i \in \lsem 1 , 4\cdot |\Agt|\rsem$, $v^f_i(s,D) = v_i(s,D)$.
\item if $(s,D') \in \Succ(D)$ and \eve can ensure $u$ from $(s,D')$, then for all $i \in \lsem 1 , 4 \cdot |\Agt|\rsem$, $v^f_i(s,D') = W$.
\item if $(s,D') \in \Succ(D)$ and \eve cannot ensure $u$ from $(s,D')$, then  for all $i \in \lsem 1 , 4 \cdot |\Agt|\rsem$, $v^f_i(s,D') = -W-1$.
\end{inparaenum}

\begin{collect*}{appendix-fixed-coalition}{
\begin{lemma}\mylabel{Lem}{lem:fixed-coalition} 
  \eve can ensure payoff $u \in \lsem -W, W\rsem^d$ in $\devgame$ from $(s,D)$ if, and only if, she can ensure $u$ in the fixed coalition game $\mathcal{F}(D,p)$ from $(s,D)$.
\end{lemma}
}{}{}{
\begin{proof}
  \fbox{$\Rightarrow$}
  Let $\sigma_\exists$ be strategy which ensures $u$ in $\devgame$ from $(s,D)$.
  If we apply the strategy in $\mathcal{F}$, any of its outcome $\rho$ from $(s,D)$ will either stay in the $D$ component and correspond to an outcome of $\sigma_\exists$ in $\devgame$ or reach a state $(s',D')$ in $\Succ(D)$.
  Since \eve can ensure the payoff $u$ in $\devgame$, and $(s',D')$ is reached by one of its outcome, she can do so from $(s',D')$.
  Therefore $(s',D')$ is a state where the weights are maximal (equal to $W$ on all dimensions) and the outcome is winning in $\mathcal{F}(D,u)$.

  \fbox{$\Leftarrow$}
  Let $\sigma_\exists$ be strategy that ensures $u$ in $\mathcal{F}(D,p)$.
  Every outcome of this strategy that get out of the $D$ component, reach a state $(s',D')\in \Succ(D)$ where the weights are greater than $u$.
  These states cannot be losing since the weights on any dimension $i$ would be smaller than $-W-1$, which is smaller than $u_i$.
  This means that these states are winning, and to each state $(s',D')$ with $D' \ne D$ that is reached by an outcome of $\sigma_\exists$ we can associate a strategy $\sigma^{s',D'}_\Eve$ that is winning from $(s',D')$.
  We consider the strategy $\sigma'_\exists$ that plays according to $\sigma_\exists$ as long has we stay in the $D$ component and according to $\sigma^{s',D'}_\Eve$ once we reach another component, where $(s',D')$ is the first state outside the $D$ component that was reached.

  The construction is such that the strategy $\sigma'_\exists$ ensures $u$ in $\devgame$:
  let $\rho$ be one of its outcome, if $\rho$ stays in the $D$ component, it is also an outcome of $\sigma_\exists$ with the same payoff which is above $u$; otherwise $\rho$ reaches a state $(s',D')$ in $\Succ(D)$ and from this point \eve follows a strategy that ensures threshold $u$.
\end{proof}
}
\end{collect*}

Using this correspondence, we deduce a polynomial space algorithm to check that \eve can ensure a given value in the deviator game and thus a polynomial space algorithm for the robustness problem.
\begin{theorem}\mylabel{Thm}{thm:pspace-membership}
  There is a polynomial space algorithm, that given a concurrent game $\calG$, tells if there is a $(k,t,r)$-robust equilibrium.
\end{theorem}
\begin{proof}
  We first show that there is a polynomial space algorithm to solve the value problem in $\devgame$.
  We consider a threshold $u$ and a state $(s,D)$.
  In the fixed coalition game~$\mathcal{F}(D,u)$, for each $(s',D') \in \Succ(D)$, we can compute whether it is winning by recursive calls.
  Once the weights for all $(s',D') \in \Succ(D)$ have been computed for $\mathcal{F}(D,u)$, we can solve the value problem in $\mathcal{F}(D,u)$.
  Thanks to Lem.~\ref{lem:fixed-coalition} the answer to value problem in this game is yes exactly when \eve can ensure $u$ from $(s,D)$ in $\devgame$.
  There is a \co\NP~algorithm~\cite{velner12} to check the value problem in a given game, and therefore there also is an algorithm which uses polynomial space.
  The size of the stack of recursive calls is bounded by $|\Agt|$, so the global algorithm uses polynomial space.

  We now use this to show that there is a polynomial space algorithm for the polyhedron value problem in $\devgame$.
  We showed in Lem.~\ref{lem:small-witness}, that if the polyhedron value problem has a solution then there is a threshold $u$ of polynomial size that is witness of this property.
  We can enumerate all the thresholds that satisfy the size bound in polynomial space.
  We can then test that these thresholds satisfy the given linear inequations, and that the algorithm for the value problem answers yes on this input, in polynomial space thanks to the previous algorithm.
  If this is the case for one of the thresholds, then we answer yes for the polyhedron value problem.
  The correctness of this procedure holds thanks to Lem.~\ref{lem:small-witness}.

  We now use this to show that there is a polynomial space algorithm for the robustness problem.
  Given a game $\calG$ and parameters $(k,t,r)$, we define a tuple of linear equations, 
  for all $i \in \lsem 1 , |\Agt| \rsem$, \(x_{2\cdot |\Agt|+i} = x_i + r \land x_{2\cdot |\Agt|+i} = - x_{|\Agt|+i}
  \land x_{2\cdot |\Agt|+i} = - x_{3 \cdot |\Agt|+i} \)
  (each equation can be expressed by two inequations).
  Thanks to Lem.~\ref{lem:correct-mean-robust}, there is a payoff which satisfies these constraints and which \eve can ensure in $\devgame$ if, and only if, there is a $(k,t,r)$-robust equilibrium.
  Then, querying the algorithm we described for the polyhedron value problem in $\devgame$ with our system of inequations, answers the robustness problem.
\end{proof}

\section{Hardness}\label{sec:hardness}
In this section, we show a matching lower bound for the resilience problem.
The lower bound holds for weights that are $0$ in every states except on some terminal states where they can be $1$.
This is also called \emph{simple reachability objectives}.

\begin{collect*}{appendix-hardness}{\begin{theorem}
  The robustness problem is \PSPACE-complete.
\end{theorem}}{}{}{}
\end{collect*}

Note that we already proved \PSPACE-membership in Thm.~\ref{thm:pspace-membership}.
We give the construction and intuition of the reduction to show hardness and leave the proof of correctness in the appendix.
We encode \QSAT formulas with $n$ variable into a game with $2\cdot n+2$ players, such that the formula is valid if, and only if, there is $n$-resilient equilibria.
  We assume that we are given a formula of the form $\phi = \forall x_1. \exists x_2.\ \forall x_3 \cdot \exists x_n.\ C_1 \land \cdots \land C_k$, where each $C_k$ is of the form $\ell_{1,k} \lor \ell_{2,k} \lor \ell_{3,k}$ and each $\ell_{j,k}$ is a literal (i.e. $x_m$ or $\lnot x_m$ for some $m$).
  We define the game $\calG_\phi$ as illustrated by an example in \figurename~\ref{fig:hardness}.
  It has a player $A_m$ for each positive literal $x_m$, and a player $B_m$ for each negative literal $\lnot x_m$.
  We add two extra players \eve and \adam.
  \eve is making choices for the existential quantification and \adam for the universal ones.
  When they chose a literal, the corresponding player can either go to a sink state $\bot$ or continue the game to the next quantification.
  Once a literal has been chosen for all the variables, \eve needs to chose a literal for each clause.
  The objective for \eve and the literal players is to reach $\bot$.
  The objective for \adam is to reach $\top$.
  We ask whether there is a $(n+1)$-resilient equilibrium.

  To a history $h=\adam_1 \cdot X_1\cdot \eve_2 \cdot X_2 \cdot \adam_3 \cdots \eve_m \cdot X_m$ with $X_i \in \{ A_i,B_i\}$, we associate a valuation $v_h$, such that $v_h(x_i) = \true$ if $X_i = B_i$ and $v_h(x_i) = \false$ if $X_i = A_i$.
  Intuitively, \eve has to find a valuation that makes the formula hold, while \adam tries to falsify it.

\begin{collect}{appendix-hardness}{}{}
\begin{proof}
  Note first that if the outcome is going to the state winning for \adam, it is possible for a $A_i$ to change its strategy and go to $\bot$, thus improving its payoff.
  Therefore a $(n+1)$-resilient equilibrium is necessarily losing for \adam and winning for all the others.

  \fbox{Validity $\Rightarrow$ equilibrium.}
  Assume that $\phi$ is valid, we will show that there is a $(n+1)$-resilient equilibrium.
  We define a strategy of \eve such that if $v_h$ makes $\exists x_m.\ \forall x_{m+1} \cdots \exists x_{n}.\ C_1 \land \cdots \land C_k$ valid, then $\sigma_\exists(h) = X_{m}$ such that $v_{h\cdot X_m}$ makes $\forall x_{m+1} \cdots \exists x_{n}.\ C_1 \land \cdots \land C_k$ valid.
  As $\phi$ is valid, we know that for all outcomes~$h$ of $\sigma_\exists$ of the form $\adam_1 \cdot X_1\cdots \eve_k \cdot X_k$, $v_h$ makes $C_1 \land \cdots \land C_k$ valid.
  Then from $X_m$, \eve can choose for each close a state $Y$ that is different from all $X_1 \dots X_m$.
  We also fix the strategy of all players $A_i$ and $B_i$ and \adam to go to the state $\bot$.
  This defines a strategy profile that we will write $\sigma_\Agt$.

  Consider a strategy profile $\sigma'_\Agt$ where at most $(n+1)$ strategies are different from the ones in $\sigma_\Agt$. 
  Assume $\sigma'_\Agt$ reaches $\eve_m$.
  We know that in $\sigma'_\Agt$ at least $n+1$ strategies are different from the ones $\sigma_\Agt$ and $\{ A \in \Agt \mid \sigma'_A \ne \sigma_A\} = \{ \adam , X_1 , \dots , X_m\}$.
  Then, by the choice of the strategy for \eve, the states that are seen in the following are controlled by players that are different from $X_1,\dots,X_m$.
  Thus the run ends in $\bot$.

  \medskip
  \fbox{Equilibrium $\implies$ validity.}
  Assume that $\sigma_\shortEve$ is part of a $(n+1)$-resilient equilibrium, we will show that $\phi$ is valid.
  Given a partial valuation $v_m \colon \{ x_1 , \dots , x_m\} \mapsto \{ \true, \false\}$, we define the function $f(v_m)$ such that:
  \[f(v_m) \Leftrightarrow \sigma_\exists(\adam_1 \cdot X_1 \cdots \adam_m \cdot X_m) = B_{m+1}.\]
  We will show that every valuation $v$, such that $v(x_{2k}) = f (v_{|2k-1})$, makes the formula $C_1 \land \cdots \land C_k$ valid, which shows that the formula~$\phi$ is valid.
  
  For all such valuations $v$, we can define strategies of \adam and players $X_i$ such that $X_i = A_i$ if $v(x_i) = \false$ and $X_i = B_i$ otherwise, such that keeping all other strategies similar to $\sigma_\Agt$, the state $\eve_m$ is reached.
  Then, if we see a state belonging to one of the $X_i$, we can make the strategy go to the $\top$ state.
  Since the profile is $(n+1)$-resilient, this is impossible.
  Which shows that $\sigma_\eve$ chooses for each clause a literal such that $v(\ell) = \true$.
  Therefore $v$ makes the formula $C_1 \land \cdots \land C_k$ valid.
\end{proof}
\end{collect}
  \begin{figure}[thb]
    \centering{
      \begin{tikzpicture}[yscale=0.6]
        \everymath{\scriptstyle}
        \draw (0,0) node[draw,rounded corners=2mm] (I1) {$\Adam$};
        \draw (1,1) node[draw,rounded corners=2mm] (A1) {$A_1$};
        \draw (1,-1) node[draw,rounded corners=2mm] (B1) {$B_1$};
        \draw (1,0) node[draw,rounded corners=2mm,dashed] (C1) {$\bot$};

        \draw (2,0) node[draw,rounded corners=2mm] (I2) {$\Eve$};
        \draw (3,1) node[draw,rounded corners=2mm] (A2) {$A_2$};
        \draw (3,-1) node[draw,rounded corners=2mm] (B2) {$B_2$};
        \draw (3,0) node[draw,rounded corners=2mm,dashed] (C2) {$\bot$};

        \draw (4,0) node[draw,rounded corners=2mm] (I3) {$\Adam$};
        \draw (5,1) node[draw,rounded corners=2mm] (A3) {$A_3$};
        \draw (5,-1) node[draw,rounded corners=2mm] (B3) {$B_3$};
        \draw (5,0) node[draw,rounded corners=2mm,dashed] (C3) {$\bot$};

        \draw (6,0) node[draw,rounded corners=2mm] (I4) {$\Eve$};
        \draw (7,1) node[draw,rounded corners=2mm] (A4) {$A_4$};
        \draw (7,-1) node[draw,rounded corners=2mm] (B4) {$B_4$};
        \draw (7,0) node[draw,rounded corners=2mm,dashed] (C4) {$\bot$};

        \draw (8,0) node[draw,rounded corners=2mm] (I5) {$\Eve$};
        \draw (9,1) node[draw,rounded corners=2mm] (CA1) {$A_1$};
        \draw (9,0) node[draw,rounded corners=2mm] (CA2) {$A_2$};
        \draw (9,-1) node[draw,rounded corners=2mm] (CB3) {$B_3$};
        \draw (10,1) node[draw,rounded corners=2mm,dashed] (C5) {$\top$};
        \draw (10,-1) node[draw,rounded corners=2mm,dashed] (C6) {$\top$};

        \draw (10,0) node[draw,rounded corners=2mm] (I6) {$\Eve$};
        \draw (11,1) node[draw,rounded corners=2mm] (DB2) {$B_2$};
        \draw (11,0) node[draw,rounded corners=2mm] (DA3) {$A_3$};
        \draw (11,-1) node[draw,rounded corners=2mm] (DA4) {$A_4$};

        \draw (12,0) node[draw,rounded corners=2mm,dashed] (END) {$\bot$};

        \draw[-latex'] (I1) -- (A1);
        \draw[-latex'] (I1) -- (B1);
        \draw[-latex',dashed] (I1) -- (C1);
        \draw[-latex'] (A1) -- (I2);
        \draw[-latex'] (B1) -- (I2);
        \draw[-latex',dashed] (A1) -- (C1);
        \draw[-latex',dashed] (B1) -- (C1);

        \draw[-latex'] (I2) -- (A2);
        \draw[-latex'] (I2) -- (B2);
        \draw[-latex'] (A2) -- (I3);
        \draw[-latex'] (B2) -- (I3);
        \draw[-latex',dashed] (A2) -- (C2);
        \draw[-latex',dashed] (B2) -- (C2);
        
        \draw[-latex'] (I3) -- (A3);
        \draw[-latex'] (I3) -- (B3);
        \draw[-latex'] (A3) -- (I4);
        \draw[-latex'] (B3) -- (I4);
        \draw[-latex',dashed] (A3) -- (C3);
        \draw[-latex',dashed] (B3) -- (C3);

        \draw[-latex'] (I4) -- (A4);
        \draw[-latex'] (I4) -- (B4);
        \draw[-latex'] (A4) -- (I5);
        \draw[-latex'] (B4) -- (I5);
        \draw[-latex',dashed] (A4) -- (C4);
        \draw[-latex',dashed] (B4) -- (C4);

        \draw[-latex'] (I5) -- (CA1);
        \draw[-latex'] (I5) -- (CA2);
        \draw[-latex'] (I5) -- (CB3);
        \draw[-latex',dashed] (CA1) -- (I6);
        \draw[-latex',dashed] (CA2) -- (I6);
        \draw[-latex',dashed] (CB3) -- (I6);
        \draw[-latex'] (CA1) -- (C5);
        \draw (CA2) edge[-latex'] (C5.-100);
        \draw[-latex'] (CB3) -- (C6);

        \draw[-latex'] (I6) -- (DB2);
        \draw[-latex'] (I6) -- (DA3);
        \draw[-latex'] (I6) -- (DA4);
        \draw[-latex'] (DB2) -- (C5);
        \draw (DA3) edge[-latex'] (C5.-80);
        \draw[-latex'] (DA4) -- (C6);

        \draw[-latex',dashed] (DB2) -- (END);
        \draw[-latex',dashed] (DA3) -- (END);
        \draw[-latex',dashed] (DA4) -- (END);

      \end{tikzpicture}
      \caption{Encoding of a formula  $\phi = \forall x_1. \exists x_2.\ \forall x_3.\ \exists x_4.\ (x_1 \lor x_2 \lor \lnot x_3) \land (\lnot x_2 \lor x_3 \lor x_4)$.
        The dashed edges represent the strategies in the equilibrium of the players other than \eve.
      }\label{fig:hardness}
    }
  \end{figure}


\newpage

\bibliography{bib}

\begin{thebibliography}{10}

\bibitem{abraham2006distributed}
I.~Abraham, D.~Dolev, R.~Gonen, and J.~Halpern.
\newblock Distributed computing meets game theory: robust mechanisms for
  rational secret sharing and multiparty computation.
\newblock In {\em Proceedings of the twenty-fifth annual ACM symposium on
  Principles of distributed computing}, pages 53--62. ACM, 2006.

\bibitem{AHK02}
R.~Alur, T.~A. Henzinger, and O.~Kupferman.
\newblock Alternating-time temporal logic.
\newblock {\em Journal of the ACM (JACM)}, 49(5):672--713, 2002.

\bibitem{aumann1959acceptable}
R.~Aumann.
\newblock Acceptable points in general cooperative n-person games.
\newblock {\em Topics in Mathematical Economics and Game Theory Essays in Honor
  of Robert J Aumann}, 23:287--324, 1959.

\bibitem{BBMU12}
P.~Bouyer, R.~Brenguier, N.~Markey, and M.~Ummels.
\newblock Concurrent games with ordered objectives.
\newblock In L.~Birkedal, editor, {\em {FoSSaCS}'12}, volume 7213 of {\em
  {LNCS}}, pages 301--315. Springer-Verlag, Mar. 2012.

\bibitem{BMS14}
P.~Bouyer, N.~Markey, and D.~Stan.
\newblock {Mixed Nash Equilibria in Concurrent Terminal-Reward Games}.
\newblock In {\em 34th International Conference on Foundation of Software
  Technology and Theoretical Computer Science (FSTTCS 2014)}, volume~29 of {\em
  Leibniz International Proceedings in Informatics (LIPIcs)}, pages 351--363,
  Dagstuhl, Germany, 2014.

\bibitem{brenguier12}
R.~Brenguier.
\newblock {\em Nash equilibria in concurrent games: application to timed
  games}.
\newblock PhD thesis, Cachan, Ecole normale sup{\'e}rieure, 2012.

\bibitem{long-version}
R.~Brenguier.
\newblock Robust equilibria in concurrent games.
\newblock {\em CoRR}, abs/1311.7683, 2015.

\bibitem{cav15-long}
R.~Brenguier and J.-F. Raskin.
\newblock Optimal values of multidimensional mean-payoff games, 2014.
\newblock https://hal.archives-ouvertes.fr/hal-00977352/.

\bibitem{BR15}
R.~Brenguier and J.-F. Raskin.
\newblock Pareto curves of multidimensional mean-payoff games.
\newblock In {\em Computer Aided Verification}, pages 251--267. Springer, 2015.

\bibitem{brihaye2010}
T.~Brihaye, V.~Bruy\`ere, and J.~De~Pril.
\newblock Equilibria in quantitative reachability games.
\newblock In F.~Ablayev and E.~Mayr, editors, {\em Computer Science – Theory
  and Applications}, volume 6072 of {\em {LNCS}}, pages 72--83. Springer Berlin
  / Heidelberg, 2010.

\bibitem{Can88}
J.~Canny.
\newblock Some algebraic and geometric computations in pspace.
\newblock In {\em Proceedings of the twentieth annual ACM symposium on Theory
  of computing}, pages 460--467. ACM, 1988.

\bibitem{CHJ05}
K.~Chatterjee, T.~A. Henzinger, and M.~Jurdzi{\'n}ski.
\newblock Games with secure equilibria.
\newblock In {\em Formal Methods for Components and Objects}, pages 141--161.
  Springer, 2005.

\bibitem{CHP10}
K.~Chatterjee, T.~A. Henzinger, and N.~Piterman.
\newblock Strategy logic.
\newblock {\em Information and Computation}, 208(6):677--693, 2010.

\bibitem{MMPV12}
F.~Mogavero, A.~Murano, G.~Perelli, and M.~Y. Vardi.
\newblock What makes {ATL}* decidable? a decidable fragment of strategy logic.
\newblock In {\em CONCUR 2012--Concurrency Theory}, pages 193--208. Springer,
  2012.

\bibitem{nash50}
J.~F. Nash, Jr.
\newblock Equilibrium points in {\(n\)}-person games.
\newblock {\em Proc.\ National Academy of Sciences of the~{USA}}, 36(1):48--49,
  Jan. 1950.

\bibitem{Puterman94}
M.~L. Puterman.
\newblock Markov decision processes: Discrete stochastic dynamic programming.
\newblock 1994.

\bibitem{Ummels08}
M.~Ummels.
\newblock The complexity of {N}ash equilibria in infinite multiplayer games.
\newblock In {\em Proc. of {FoSSaCS'12}}, volume 4962 of {\em {LNCS}}, pages
  20--34. {Springer}, 2008.

\bibitem{ummels2009complexity}
M.~Ummels and D.~Wojtczak.
\newblock The complexity of {N}ash equilibria in simple stochastic multiplayer
  games.
\newblock {\em Automata, Languages and Programming}, pages 297--308, 2009.

\bibitem{ummels2011complexity}
M.~Ummels and D.~Wojtczak.
\newblock The complexity of {N}ash equilibria in limit-average games.
\newblock {\em CONCUR 2011--Concurrency Theory}, pages 482--496, 2011.

\bibitem{velner12}
Y.~Velner, K.~Chatterjee, L.~Doyen, T.~A. Henzinger, A.~Rabinovich, and J.-F.
  Raskin.
\newblock The complexity of multi-mean-payoff and multi-energy games.
\newblock {\em CoRR}, abs/1209.3234, 2012.

\end{thebibliography}

\newpage
\renewcommand\mylabel[2]{(#1.~\ref{#2} in the body of the paper)}

\appendix

\section{Appendix for Section~3}
\includecollection{appendix-randomised}

\section{Appendix for Section~4}
\includecollection{appendix-suspect}

\section{Appendix for Section~5}
\includecollection{appendix-mean-payoff}

\section{Appendix for Section~6}
\includecollection{appendix-fixed-coalition}

\section{Appendix for Section~7}
\includecollection{appendix-hardness}

\end{document}